\documentclass[a4paper,12pt,reqno]{amsart}
\usepackage{amsmath,amsthm,amsfonts,amssymb}
\usepackage[T1]{fontenc}
\usepackage{times}
\usepackage{paralist}
\usepackage[pdftex]{graphicx}
\usepackage{enumerate}

\usepackage{bbm}
\usepackage{color}
\usepackage{graphicx}
\usepackage{multirow}
\usepackage{subcaption}

\setdefaultenum{(i)}{}{}{}
\usepackage[round]{natbib}
\usepackage[margin=1.2in]{geometry}

\long\def\symbolfootnote[#1]#2{\begingroup\def\thefootnote{\fnsymbol{footnote}}\footnote[#1]{#2}\endgroup}

\newcommand{\cA}{\mathcal{A}}

 \newcommand{\cD}{\mathcal{D}}

\newcommand{\cF}{\mathcal{F}}

\newcommand{\cH}{\mathcal{H}}

\newcommand{\cJ}{\mathcal{J}}

\newcommand{\cL}{\mathcal{L}}
\newcommand{\cM}{\mathcal{M}}

 \newcommand{\cS}{\mathcal{S}}

\newcommand{\pa}{{\partial}}
\newcommand{\ignore}[1]{}
\def \a{\alpha}
\def \b{\beta}
\def \e{\epsilon}
\def \d{\delta}

\def \E{\mathbb{E}}

 \def\P{\mathbb{P}}

\def \R{\mathbb{R}}

\def \Sb{\mathbb {S}}

\def\1{{\bf{1}}}
\def\0{{\bf{0}}}

\theoremstyle{plain}
\newtheorem{thm}{Theorem}[section]
\newtheorem{dfn}[thm]{Definition}
\newtheorem{prop}[thm]{Proposition}
\newtheorem{lem}[thm]{Lemma}

\newtheorem{asm}[thm]{Assumption} 

\theoremstyle{definition}
\newtheorem{rem}[thm]{Remark} 
 
\newtheorem{example}[thm]{Example}

\numberwithin{equation}{section}

\renewenvironment{thebibliography}[1]{%
\begin{oldthebibliography}{#1}%
\setlength{\baselineskip}{.9em}
\linespread{.9}
\small
\setlength{\parskip}{0ex}%
\setlength{\itemsep}{.1em}%
}%
{%
\end{oldthebibliography}%
}

\title{Portfolio Choice with Small Temporary and Transient Price Impact}

\begin{document}
\symbolfootnote[0]{The insightful remarks of two anonymous reviewers are gratefully acknowledged.}
\symbolfootnote[0]{$\dag$ Florida State University, Department of Mathematics, 224 LOVE Building, 1017 Academic Way, Tallahassee, FL 32306, USA, email \texttt{iekren@fsu.edu}. Partly supported by Swiss National Science Foundation Grant SNF 153555.}
\symbolfootnote[0]{$\ddag$Carnegie Mellon University, Department of Mathematical Sciences, 5000 Forbes Avenue, Pittsburgh, PA 15213, USA, email: \texttt{johannesmk@cmu.edu}.}

\maketitle

\vspace{-.5cm}
\centerline{\textsc{Ibrahim Ekren$^{\dag}$}}
\medskip\centerline{\textit{Florida State University}}
\medskip\medskip

\centerline{\textsc{Johannes Muhle-Karbe$^{\ddag}$}}
\medskip\centerline{\textit{Carnegie Mellon University}}
\medskip\medskip

We study portfolio selection in a model with both temporary and transient price impact introduced by \cite{garleanu.pedersen.16}. In the large-liquidity limit where both frictions are small, we derive explicit formulas for the asymptotically optimal trading rate and the corresponding minimal leading-order performance loss. We find that the losses are governed by the volatility of the frictionless target strategy, like in models with only temporary price impact. In contrast, the corresponding optimal portfolio not only tracks the frictionless optimizer, but also exploits the displacement of the market price from its\ unaffected level.



\medskip\medskip\medskip
\normalsize

\section{Introduction}

When rebalancing large portfolios, the adverse price impact of each trade is a key concern. Indeed, large transactions deplete the liquidity available in the market and lead to less favorable execution prices. After the completion of a large trade liquidity recovers, but only gradually. Whence, it is of crucial importance for portfolio managers to schedule their order flow in an appropriate manner, so as to trade off the gains and costs of rebalancing in an optimal manner. 

Accordingly, there is a large and growing literature on price impact models. Following the survey paper of \cite{gatheral.schied.13}, these models can be broadly classified into two categories. The first distinguishes between \emph{temporary} trading costs, that only affect each trade separately, and \emph{permanent} price impact that affects the current and all future trades in the same manner (cf., e.g., \cite{bertsimas.lo.98,almgren.chriss.01} and many more recent studies). The second takes into account the \emph{transient} nature of price impact, which is caused by large trades but gradually wears off once these are completed, cf., e.g., \cite{bouchaud.al.04,bouchaud.al.06,obizhaeva.wang.13,gatheral.10,alfonsi.al.10,shreve.al.11,gatheral.schied.13}.

These models were originally developed for optimal execution problems, where the goal is to split up a single, exogenously given order in an optimal manner. More recently, dynamic portfolio choice and hedging problems with price impact have also received increasing attention \citep{garleanu.pedersen.13,garleanu.pedersen.16,dufresne.al.12,guasoni.weber.15a,moreau.al.15,gueant.pu.15,almgren.li.16,bank.al.16}. This means that the target orders to be executed are no longer assumed to be given, but are instead derived endogenously from a dynamic optimization problem. This allows to explicitly model the tradeoff between gains from reacting to new information and costs of trading. However, the complexity of the optimization problem increases considerably. Whence, attention has almost exclusively focused on first-generation price impact models with only temporary trading costs so far. 

The only exception is the recent work of \cite{garleanu.pedersen.16}. They study portfolio choice for agents that try to exploit partially predictable returns in the presence of linear temporary \emph{and} transient price impact. Using dynamic programming arguments, they describe the value function of the problem at hand and the corresponding optimal trading rate via the solution of a coupled system of nonlinear equations.\footnote{The coupled nature of these optimality equations complicates the analysis of even the simplest concrete models, unlike for models with purely temporary costs, where linear-quadratic problems can be solved explicitly in essentially full generality \citep{garleanu.pedersen.13,garleanu.pedersen.16,cartea.jaimungal.16,bank.al.16,bank.voss.17}.}  This analysis identifies the current deviation of the market price from its ``unaffected'' value as an important new state variable. However, the involved nonlinear nature of the optimality conditions makes it difficult to draw further qualitative and quantitative conclusions beyond a benchmark model with a linear factor process.

To overcome this lack of tractability and analyze models with more general dynamics, small-cost asymptotics have proven to be very useful in models with temporary trading costs only. This means that one views the trading friction at hand as a perturbation of the frictionless benchmark model, and looks for corrections of the frictionless optimizer that take it into account in an asymptotically optimal manner. As succinctly summarized by \cite{whalley.wilmott.97}, the goal is to ``reveal the salient features of the problem while remaining a good approximation to the full but more complicated model''. For example, in the context of linear temporary price impact, \cite*{moreau.al.15} have shown that both the optimal trading rate and the leading-order loss due to transaction costs admit explicit asymptotic expressions.\footnote{Related work on other small transaction costs includes~\cite{shreve.soner.94,whalley.wilmott.97,korn.98,janecek.shreve.04,bichuch.14,soner.touzi.13,possamai.al.15,martin.12,kallsen.muhlekarbe.15,kallsen.li.15,altarovici.al.15,cai.al.15,cai.al.16,feodoria.16,melnyk.seifried.17,herdegen.muhlekarbe.17}.} The trading rate turns out to be proportional to the distortion relative to the frictionless target and a universal constant -- the square-root of risk aversion, times market variance, divided by trading costs.\footnote{The same statistic also plays a key role in optimal execution problems \cite{almgren.chriss.01,schied.schoeneborn.09} and models with asymmetric information \cite{muhlekarbe.webster.16}.} In contrast, the volatility of the target strategy does not feature in this formula, so that the optimal relative trading speed is the same for a broad class of optimization problems. In contrast the volatility of the frictionless target is a crucial input for the leading-order effect of the trading costs, which are equal to a suitably weighted average of this quantity, weighted with a term explicitly determined by risk aversion, market volatility, and trading costs.

 \begin{figure}
  \centering
  \includegraphics[width=.7\linewidth]{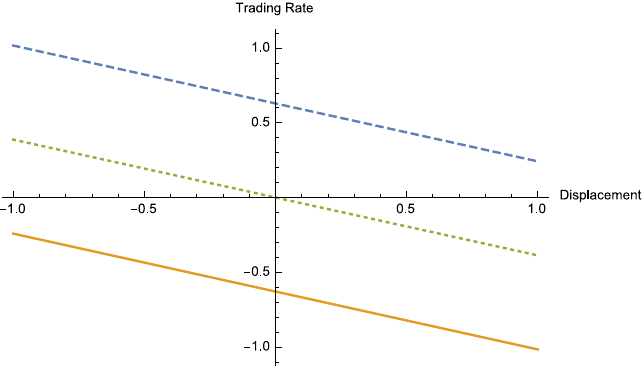}
\caption{Asymptotically optimal trading speed in shares per unit of time plotted against deviation from frictionless optimum, for zero price distortion (dotted), negative distortion (dashed), and positive distortion (solid) with $\Sigma=0.3$, $\Lambda=1$, $C=0.5$ and $\gamma=R=1/2$.}
\label{fig:test2}
\end{figure}

The present study brings similar asymptotic methods to bear on the model of \cite{garleanu.pedersen.16}. As this model includes two frictions -- temporary and transient price impact -- we consider the joint limit where both of these become small. In order to understand the contribution of both frictions, we focus on the ``critical regime'', where both are rescaled so as to feature nontrivially in the limit. 

In a general Markovian setting, this enables us to obtain tractable formulas for the asymptotically optimal trading rate and the leading-order performance loss due to illiquidity due to both temporary and transient price impact. These results show that the volatility of the frictionless target portfolio is still the key statistic for its sensitivity with respect to small trading frictions. Indeed, the representation of the first-oder loss from \cite{moreau.al.15} remains valid in the present context after updating the scaling weight to account for the additional model parameters. In contrast, the optimal trading rate is more complex. To wit, with only temporary transaction costs, the asymptotically optimal policy simply tracks the frictionless target. The transient price distortion in the present model provides an additional predictor for future price changes. Accordingly, the optimal trading rate now trades off tracking the frictionless target against the exploitation of this trading signal. 

This is displayed in Figure~\ref{fig:test2} for a model with one risky asset and constant optimal trading speed. As a positive distortion predicts negative future price changes, the corresponding trading rate is reduced and vice versa. The more involved comparative statics for several risky assets are illustrated in Figures~\ref{fig:test1a} and \ref{fig:test1b}. There, we plot the vector field of asymptotically optimal trading speeds for two assets and display the effect a price distortion in one of the assets has on these. In Figure~\ref{fig:test1a}, the two assets are uncorrelated. Accordingly, a positive price distortion in asset one only affects the optimal trading speed in the latter, whereas the turnover in asset two remains unchanged.\footnote{This decoupling between uncorrelated assets is typical for agents with constant absolute risk aversions \citep{liu.04}, but only holds true approximately in the high risk-aversion limit for constant relative risk aversions \citep{guasoni.muhlekarbe.15}. The comparative statics of a multidimensional model with temporary price impact and constant relative risk aversion are discussed in detail by \cite{guasoni.weber.13}.}  In Figure~\ref{fig:test1b}, the two assets have substantial positive correlation. Then, a positive price distortion in the first asset still is a negative indicator for its future returns, so that the corresponding target position is reduced. However, to compensate for this change of the portfolio composition, the target position in the positively correlated second asset now is increased.

\begin{figure}
\centering
\begin{subfigure}{.5\textwidth}
  \centering
  \includegraphics[width=.9\linewidth]{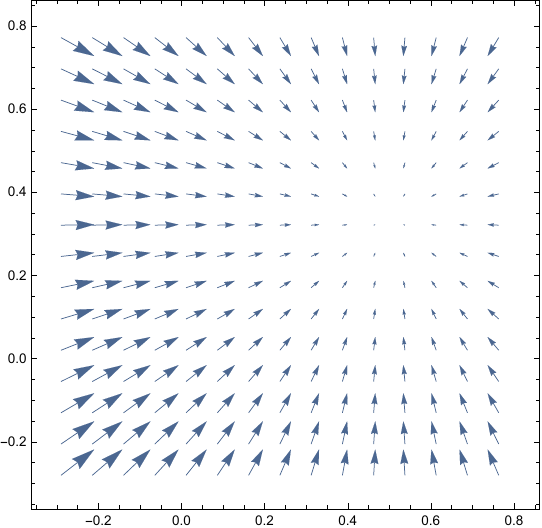}
\end{subfigure}%
\begin{subfigure}{.5\textwidth}
  \centering
  \includegraphics[width=.9\linewidth]{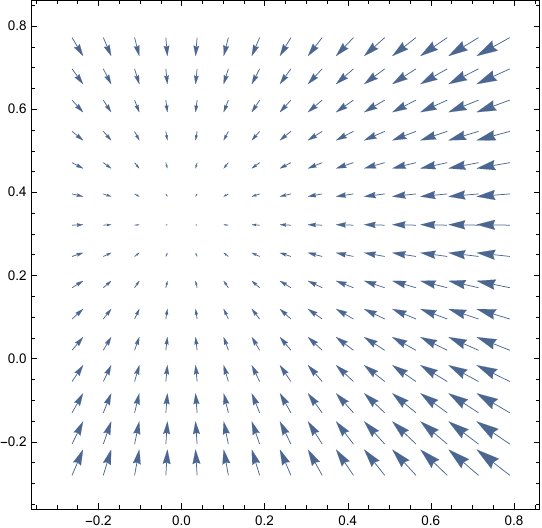}
\end{subfigure}
\caption{Optimal trading rates in two uncorrelated assets, with zero price distortion (left panel), and positive price distortion in asset one (right panel). Parameters are $\Sigma=\tiny\protect\begin{pmatrix} 1 & 0 \\ 0 & 2 \protect\end{pmatrix}$, $\Lambda=\Sigma/2$, $C=2\Sigma$ and $\gamma=R=1/2$.}
\label{fig:test1a}
\end{figure}

\begin{figure}
\centering
\begin{subfigure}{.5\textwidth}
  \centering
  \includegraphics[width=.9\linewidth]{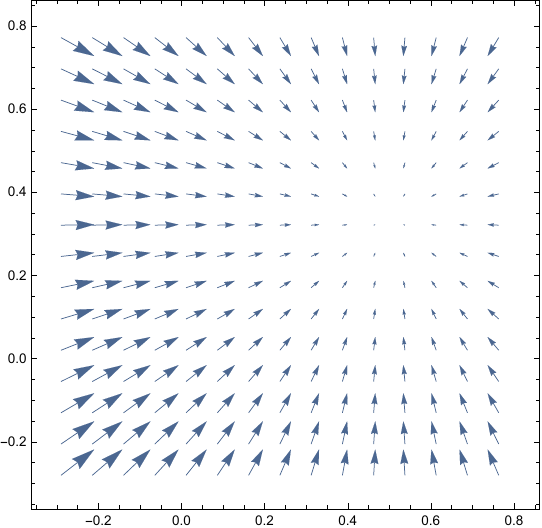}
\end{subfigure}%
\begin{subfigure}{.5\textwidth}
  \centering
  \includegraphics[width=.9\linewidth]{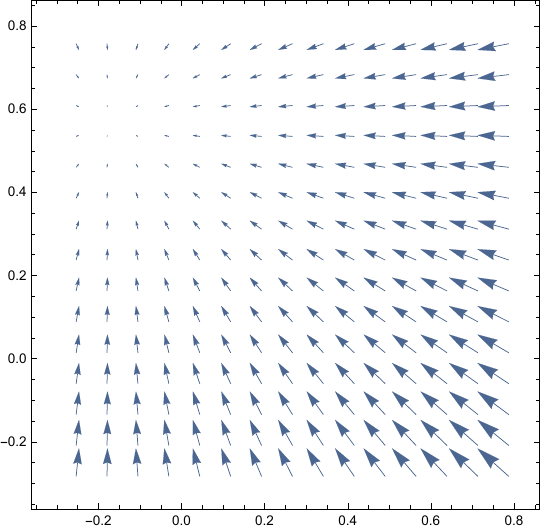}
\end{subfigure}
\caption{Optimal trading rates in two positively correlated assets, with zero price distortion (left panel), and positive price distortion in asset one (right panel). Parameters are $\Sigma=\tiny\protect\begin{pmatrix} 1 & \frac{\sqrt{2}}{2} \\ \frac{\sqrt{2}}{2} & 2 \protect\end{pmatrix}$, $\Lambda=\Sigma/2$, $C=2\Sigma$ and $\gamma=R=1/2$.}
\label{fig:test1b}
\end{figure}

These effects depicted in these figures are clearly visible in our simple asymptotic formulas. More generally, these show that the tracking speed for the frictionless target turns out to be the same as for purely temporary costs. In contrast, the weight placed on the current distortion depends on the tradeoff between all price impact parameters. If temporary price impact is substantially larger than its transient counterpart, this dependence reduces to the simple ratio of price distortion over temporary trading cost, cf.~Section~\ref{ss:1d}.

To prove these results, we follow \cite{soner.touzi.13,altarovici.al.15,moreau.al.15} and use stability results for viscosity solutions. However, substantial new difficulties arise here due to the presence of an additional  state variable, the price distortion caused by trading. Its presence leads to a substantially more complicated limiting control problem. Moreover, with only temporary costs, the frictional value is always dominated by its frictionless counterpart and the partial differential equations involved have some non-degeneracy that is crucial for establishing their expansions. This is no longer the case in the present context. To overcome these difficulties, we therefore study a suitably renormalized version of the value function and develop new methods to obtain locally uniform bounds for its scaled deviation from the frictionless value. 

The remainder of this article is organized as follows. In Section~\ref{sec:model} we introduce the model, its frictionless solution, and the dynamic programming characterization of the version with temporary and transient price impact. Our main results, a first-order expansion of the frictional value function in the large liquidity limit and a corresponding asymptotically optimal policy, are presented in Section~\ref{sec:main}. In Section \ref{s.estimates}, we give estimates that enable us in Section \ref{s.semilimits} to define the upper and lower semilimits of the rescaled deviation of the value functions from the frictionless value. We in turn characterize these semilimits as viscosity semisolutions of the second corrector equation in Section~ \ref{s.corrector.eq}.  The dependence of these semilimits on the initial states is subsequently studied in Section~\ref{s.first}. Finally, the asymptotic optimality of our candidate policy is established in Section \ref{s.optimal.policy}. To ease readability, some technical proofs are delegated to Appendix~\ref{sec:app}.\\

\subsection*{Notation}
We write $I_n$ for the identity matrix on $\mathbb{R}^n$ and denote by $\Sb_n$ the cone of $\mathbb{R}^{n \times n}$-valued, symmetric, positive definite matrices. Inequalities between symmetric matrices are understood in the sense that the difference is symmetric positive semidefinite. $M^\top$ denotes the transpose of a matrix $M$, $\mathrm{Tr}(M)$ its trace, and $|M|$ its Frobenius norm. For a function $f:\xi\in\R^n\times \R^n\to\R$, we write $\pa_\xi f=(\pa_1f^\top,\pa_2f^\top)^\top\in\R^{2n}$ for the gradient of $f$. For a locally bounded function $f$, the corresponding upper and lower semicontinuous envelopes are denoted by $f^*$ and $f_*$, respectively. To simplify notation in the technical estimates, $c$ is used to denote a generic, sufficiently large positive constant that may vary from line to line.

\section{Model}\label{sec:model}

\subsection{Financial Market}
Fix a filtered probability space $(\Omega, \cF, (\cF_t)_{t\geq0}, \P)$, endowed with an $\mathbb{R}^q$-valued Brownian motion $(W_t)_{t \geq 0}$. We consider a financial market with $1+n$ assets. The first one is safe; its price is normalized to one. The other $n$ assets are risky; their prices $(S^1_t,\ldots,S^n_t)$ are given by the first $n$ components of a $\cD$-valued Markovian state process with dynamics
\begin{align}\label{eq:factor}
Y^y_t=y+\int_0^t \mu_Y(Y^y_s)ds+\int_0^t\sigma_Y(Y^y_s)dW_s,
\end{align}
where $\cD= (0,\infty)^n\times \R^m$ or $\cD= \R^n\times \R^m$. The deterministic functions $\mu_Y :\cD\to \R^{m+n}$ and $\sigma_Y:\cD\to \R^{(m+n)\times q}$ are twice continuously differentiable and Lipschitz, so that this stochastic differential equation has a unique strong solution for any initial condition $y\in \cD$. 

To ease notation, we write 
\begin{align}\label{eq:dynamics-price}
dS_t=\mu(Y^y_t)dt+\sigma(Y^y_t)dW_t,
\end{align}
and set 
$$\Sigma(y)=\sigma(y)\sigma(y)^\top,\quad \mu_t=\mu(Y^y_t),\quad \Sigma_t= \Sigma(Y^y_t),$$
when there is no ambiguity about the initial condition $y$ of $Y_t^y$. To rule out degenerate cases, we assume throughout that the infinitesimal covariance matrix $\Sigma$ is invertible with inverse $\Sigma^{-1}$. 

For all smooth functions $\phi: \cD \to \mathbb{R}$, the infinitesimal generator of the diffusion $(Y_t)_{t \geq 0}$ applied to $\phi$ is denoted by
$$
\cL^Y\phi(y):= \mu_Y^\top(y)\pa_y \phi(y)+\frac{1}{2}\mathrm{Tr}\left(\sigma_Y(y)\sigma_Y^\top(y)\pa^2_{yy}\phi(y)\right).
$$

\begin{example}\label{ex:meanrev}
Throughout the paper, we will illustrate our results for a nonlinear extension of the arithmetic model with mean-reverting returns from \cite{bouchaud.al.12,martin.12,garleanu.pedersen.13,garleanu.pedersen.16}. This means that the second component of the factor process is an autonomous Ornstein-Uhlenbeck process with dynamics
$$dY^{2,y}_t=-\lambda Y^{2,y}_tdt+\eta dW^2_t, \quad Y^{2,y}_0=y_2,$$
for $\lambda, \eta>0$. The corresponding risky asset has dynamics
$$dS^y_t=\nu (Y^{2,y}_t) dt+\sigma dW^1_t, \quad S^y_0=y_1,$$
for constants $ \sigma>0$ and $(y_1,y_2) \in \cD=\R\times\R$, as well as a function $\nu: \mathbb{R} \to \mathbb{R}$ with bounded derivatives of order one, two, and three.
\end{example}

\subsection{Trading and Optimizations without Frictions}

Without frictions, self-financing trading strategies are modeled by predictable, $\mathbb{R}^n$-valued processes $H$, where $H^i_t$ denotes the number of shares of risky asset $i=1,\ldots,n$ held at time $t$. The corresponding portfolio returns then are described by the stochastic integral $H_tdS_t$ as usual.
As in \cite{kallsen.02,grinold.06,martin.schoeneborn.11,martin.12,garleanu.pedersen.13,garleanu.pedersen.16,guasoni.mayerhofer.16}, we consider an investor with infinite planning horizon who maximizes her expected returns penalized for the corresponding variances. In the continuous-time limit this leads to the following \emph{local} mean-variance functional,
\begin{equation}\label{eq:FL}
\cJ^0(y; H):=\E\left[\int_0^\infty e^{-\rho t} \left(H^\top_t \mu(Y^y_t) -\frac{\gamma}{2} H^\top_t \Sigma(Y^y_t) H_t\right) dt\right] \to \max_H!
\end{equation}
Here, $\gamma>0$ and $\rho>0$ are the investor's risk aversion and time-discount rate, respectively. Pointwise maximization of the integrand in \eqref{eq:FL} readily yields that the optimizer is the (myopic) Merton portfolio,
\begin{align}\label{eq:markowitz}
\cM_t^y:=\cM(Y_t^y):=\frac{\Sigma(Y^y_t)^{-1}\mu(Y^y_t)}{\gamma}.
\end{align}
Hence, the value function
$$V^0(y):=\sup_{H} \cJ^0(y; H)$$
has the following probabilistic representation,
\begin{equation}\label{eq:vf}
V^0(y)=\E\left[\int_0^\infty e^{-\rho t} \frac{\mu^\top(Y^y_t)\Sigma^{-1}(Y^y_t) \mu(Y^y_t)}{2\gamma}dt\right].
\end{equation}

For our PDE analysis of the corresponding problems with frictions, we focus on the case where this value function is finite and smooth enough in the initial data to satisfy the corresponding dynamic programming equation in the classical sense.
 
\begin{asm}\label{finite-V0}
The frictionless value function \eqref{eq:vf} is finite and a classical solution of the dynamic programming equation,
\begin{align}\label{eq:V0}
\rho V^0 (y)=\cL^Y V^0(y)+\frac{\mu^\top(y)\Sigma^{-1}(y)\mu(y)}{2\gamma}, \quad\mbox{ for }y\in \cD.
\end{align}
\end{asm}

\begin{example} \label{ex:OU}
For the model with mean-reverting returns from Example~\ref{ex:meanrev}, the optimal strategy is also a (possibly nonlinear) function of the Ornstein-Uhlenbeck state process,
$$
\cM_t^y =\cM\left(Y^{2,y}_t\right)= \frac{\nu (Y^{2,y}_t)}{\gamma\sigma^2}.
$$
The corresponding value function can be computed as
\begin{align}\label{eq:frictval}V^0(y)=V^0(y^2)=\frac{1}{2\gamma\sigma^2 }\E\left[\int_0^\infty e^{-\rho t} \nu^2(Y^{2,y}_t)dt\right],
\end{align}
where $Y^{2,y}_t$ with $y=(y^1, y^2)$ is Gaussian with mean $y^2e^{-\lambda t}$ and variance $\frac{\sigma^2}{2\lambda}(1-e^{-2\lambda t})$. Given the boundedness of the first two derivatives of $\nu$, one can easily show by differentiating under the integral that $V^0$ is twice continuously differentiable. 
\end{example}

\subsection{Trading and Optimization with Frictions}

Following \cite{garleanu.pedersen.16}, we now introduce two trading frictions into the above model.\footnote{This is a special case of the general framework for transient price impact studied in \cite{gatheral.10}.} The first one is purely \emph{temporary} in that it only affects each trade separately through a quadratic cost\footnote{Put differently, each trade has a linear temporary price impact proportional to both trade size and speed, compare \cite{guasoni.weber.15,moreau.al.15} for more details.}
\begin{align}\label{eq:temp_cost}
\frac{1}{2}\dot{H}_t^{\top} \Lambda \dot{H}_t
\end{align}
levied on the turnover rate 
$$
\dot{H}_t =\frac{dH_t}{dt}.
$$
This friction is parametrized by the symmetric definite positive matrix $\Lambda\in \Sb_n$,\footnote{Symmetry can be assumed without loss of generality, compare \cite{garleanu.pedersen.13,guasoni.weber.15}; positive-definiteness means that each trade incurs a nontrivial cost. This is evidently satisfied in the most common specifications $\Lambda=\lambda I_n$ or $\Lambda=\lambda \Sigma$ for a scalar $\lambda>0$, for example.} and naturally constrains trading strategies to be absolutely continuous. As pointed out by \cite{garleanu.pedersen.16}, this ``resembles the method used by many real-world traders in electronic markets, namely to continuously post limit orders close to the best bid or ask. The trading speed is the limit orders' ``fill rate'' [\ldots]''.

In addition to the temporary trading cost, trades also have a longer-lasting impact on market prices denoted by $D$. To wit, trading at an (adapted) rate $\dot{H}_tdt$ shifts the unaffected market quote by $C\dot{H}_tdt$ for a symmetric positive definite matrix $C\in \Sb_n$, i.e., purchases create an additional positive drift, etc. However, this price impact is \emph{not} permanent as in \cite{almgren.chriss.01,almgren.li.16} but decays gradually over time with an exponential rate $R>0$.\footnote{Related models with transient price impact have been studied intensively in the optimal execution literature, compare, e.g., \cite{obizhaeva.wang.13,gatheral.10,alfonsi.al.10,shreve.al.11}.} In summary, for an initial state 
$$\theta :=(d,h,y)\in \R^n\times \R^n\times \cD$$ 
the ``transient'' distortion of the actual price from its unaffected version thus has the following Ornstein-Uhlenbeck-type dynamics,
\begin{align}\label{eq:distortion}
dD_t^{\theta,\dot H} &=-R D_t^{\theta,\dot H}dt +C\dot{H}_t dt, \quad D_0^{\theta,\dot H}=d,
\end{align}
and the risky positions $H$ evolve as 
\begin{align}\label{eq:position}
dH^{\theta,\dot H}_t=\dot H_t dt, \quad H^{\theta,\dot{H}}_0=h.
\end{align}
With temporary and transient price impact, maximizing the risk-adjusted returns of a trading strategy $H_t=\dot{H}_t dt$ then boils down to
\begin{align}
\cJ(\theta;\dot H):= \E\Bigg[\int_0^\infty e^{-\rho t}\Big((&H_t^{\theta,\dot H}) ^{\top}(\mu_t-R D_t^{\theta,\dot H}+C \dot{H}_t)-\frac{\gamma}{2} (H_t^{\theta,\dot H})^\top \Sigma_t H_t^{\theta,\dot H} \notag\\
&-\frac{1}{2}\dot{H}^\top_t \Lambda \dot{H}_t\Big) dt\Bigg] \longrightarrow \max_{\dot{H}}! \label{eq:control-e}
\end{align}
As in \cite{garleanu.pedersen.16}, this objective function means that the investor has mean-variance preferences over the changes in wealth in each time period. The first term on the right-hand side of \eqref{eq:control-e} collects the expected returns due to i) changes in the unaffected price process \eqref{eq:dynamics-price} and ii) changes in the price distortion \eqref{eq:distortion}.\footnote{The safe position pinned down by the self-financing condition does not appear explicitly because the safe asset is normalized to one.} The second term is the usual risk penalty, and the third accounts for the temporary transaction costs. Unlike for the frictionless problem \eqref{eq:FL}, the risky positions $H_t$ can no longer be adjusted immediately and for free. Instead, they become additional state variables that can only be adjusted gradually by applying the controls $\dot{H}_t$. In particular, the problem is no longer myopic and therefore needs to be attacked by dynamic programming methods. With additional transient price impact, the distortion of the current price relative to its unaffected value enters as an additional crucial statistic.

To make sure that this infinite-horizon problem \eqref{eq:control-e} is well posed, we focus on \emph{admissible strategies} that satisfy a suitable (mild) transversality condition,\footnote{\cite{garleanu.pedersen.16} mention that a condition of this type is needed, but do not provide it. If the price impact parameters are constant, our notion \eqref{eq:trans} encompasses all uniformly bounded trading rates, for example.} 
\begin{equation}\label{eq:trans}
\begin{split}
\cA_\rho(\theta):=\Bigg\{\dot{H}: &\lim_{t\to\infty }e^{-\rho t}(|H_t^{\theta,\dot H}|^2+|D_t^{\theta,\dot H}|^2)=0,\\
&\E\left[\int_0^\infty e^{-\rho t}(|H_t^{\theta,\dot H}|^2+| D_t^{\theta,\dot H}|^2) dt\right]<\infty\Bigg\}.
\end{split}
\end{equation}
(The dependence of $\cA_\rho(\theta)$ on the initial data will be omitted when there is no ambiguity.)

\subsection{Viscosity Characterization}

We now want to characterize the frictional value function
\begin{align}\label{eq:control-e2}
V(\theta)&:=\sup_{\dot{H}\in\cA_\rho(\theta)}\cJ(\theta; \dot{H}).
\end{align}
If it is locally bounded, then weak dynamic programming arguments as in \cite{bouchard.touzi.11} show that this value function is a (possibly discontinuous, compare \cite[Definition 4.2]{fleming.soner.06}) viscosity solution of the frictional dynamic programming equation:
  
 \begin{prop}\label{asm:frictional.pde}
 Suppose the frictional value function $V$ is locally bounded. Then it is a (possibly discontinuous) viscosity solution of the following frictional dynamic programming equation,
\begin{equation}\label{PDE}
\begin{split}
\rho V(\theta) =&- \frac{\gamma}{2}h^\top\Sigma(y) h+h^\top \mu(y)-R d^\top (\pa_d V(\theta)+h)+ \cL^{Y} V(\theta)\\
&+ \sup_{\dot{h}} \left\{-\frac{1}{2} \dot{h}^\top \Lambda \dot{h}+ \dot{h}^\top \pa_h V(\theta) +\dot{h}^\top C(h+\pa_d V(\theta) )  \right\},
\end{split}
\end{equation}
for all $\theta=(d,h,y)\in \R^n\times \R^n\times \cD$.
\end{prop}
  
\begin{proof}
{See Appendix~\ref{sec:app}}.
\end{proof}  

Using our transversality conditions \eqref{eq:trans}, Lemma \ref{lem:conc} in the appendix shows that the value function is indeed locally bounded under a condition on the model parameters, which sharpens \cite[Lemma 1, Equation (A.27)]{garleanu.pedersen.16}. For constant covariance matrices $\Sigma$, this condition is satisfied in particular if the discount rate $\rho$ is sufficiently small, compare Remark~\ref{rem:para}. It also holds automatically if the transient price impact $C$ is sufficiently small or the resilience parameter $R$ is sufficiently large. This applies, in particular, in the large-liquidity regime that we turn to now.

\section{Main results}\label{sec:main}

\subsection{Large-Liquidity Regime}\label{ss:lle}

Beyond linear state dynamics, the frictional dynamic programming equation \eqref{asm:frictional.pde} only allows to characterize the corresponding value function and optimal policy through the solution of a coupled system of nonlinear equations~\cite{garleanu.pedersen.16}. To shed more light on the qualitative and quantitative properties of the optimal policy and its performance, we therefore perform a large-liquidity expansion around the frictionless case. To wit, we assume that i) the temporary quadratic trading cost $\Lambda$ is small, ii) the permanent price impact $C$ is small, and iii) the mean-reversion speed $R$ towards the unaffected prices is large. 

To study how all three of these liquidity parameters influence the solution, we study the following ``critical regime'',\footnote{This scaling is chosen so that neither of the frictions dominates the other in the limit. If one of them would be sent to zero faster, then only the effects of the other would remain visible asymptotically. Similar ``matched asymptotics'' for small transaction costs and large risk aversion are studied in \cite{barles.soner.98}.} where none of them scales out as the asymptotic parameter $\e$ becomes small,\footnote{This slight abuse of notation is made to emphasize that the matrices $\Lambda$, $C$, $R$ are replaced by their rescaled versions~\eqref{eq:scaling} from now on. The scaling of the permanent impact parameter and the resilience speed are reminiscent of the high-resilience asymptotics in \cite{roch.soner.13,kallsen.muhlekarbe.14} which suggest that -- in related models with purely transient price impact -- $C/R$ should have an effect of the same order as $\Lambda$.}
\begin{align}\label{eq:scaling}
\Lambda=\e^2\Lambda,\qquad C=\e C,\qquad R=\e^{-1}R.
\end{align}

In the large-liquidity limit $\varepsilon \sim 0$, we obtain a first-order expansion of the corresponding frictional value function $V^\e$ (Theorem~\ref{thm:expansion}) and a corresponding asymptotically optimal policy (Theorem~\ref{thm:policy}). Before stating these results, we first introduce the regularity conditions  we require for our rigorous convergence proofs as well as the quantities that appear in the leading-order approximations.

\subsection{Inputs for the Expansion}\label{ss.expansionterms}

Our asymptotic expansion requires the following integrability and smoothness assumptions on the market and cost parameters, which are evidently verified for the model with mean-reverting returns from Example~\ref{ex:OU}, for example.

\begin{asm}\label{assume-expansion}
\begin{enumerate}
\item[(i)]There exists $\underline m>0$ such that for all $(\xi_1,\xi_2,y)\in \R^n\times\R^n\times D$ we have 
\begin{align}\label{eq:degenere-source}
&\frac{\gamma \xi_2^\top\Sigma (y)\xi_2}{2}+R \xi_1^\top C^{-1} \xi_1\geq \underline m|\xi|^2.
\end{align}
\item[(ii)] The following functions are locally bounded on $\cD$,
\begin{align}
\notag&M_\Sigma: y \mapsto  \int_0^\infty e^{-\rho t}\sup_{0\leq s\leq t}\E\left[1+|\Sigma_s|^4+|\cL^Y \Sigma_s|^2 +|(\sigma_Y^\top \pa_y \Sigma)_s|^{4} \right] dt,\\
&M_\cM: y \mapsto  \int_0^\infty e^{-\rho t}\sup_{0\leq s\leq t}\E\left[ 1+ |\cL^Y\cM_s|^{4}+|(\sigma_Y^\top \pa_y \cM)_s|^{4}\right] dt.\notag
\end{align}
\end{enumerate}
\end{asm}

Similarly to \cite{moreau.al.15} and \cite{soner.touzi.13}, the dependence of our expansion on the deviation from the Merton portfolio is described by the solution of the so-called ``first corrector equation'', cf.~\eqref{eq:fst}. With our quadratic costs, this equation can be solved using an algebraic Riccati equation.\footnote{Its explicit solution in the one-dimensional case is discussed in Section~\ref{ss:1d}.} Note that this solution and in turn our leading-order expansions do not depend on the discount rate $\rho$. This time-preference parameter therefore becomes negligible for the asymptotically optimal trading rates in Theorem~\ref{thm:expansion} and only appears as a scaling parameter in the value expansion from Theorem~\ref{thm:policy}, compare Lemma~\ref{lem:feynmann}.

\begin{lem}\label{lem:ricatti}
Suppose Assumption \ref{assume-expansion} is satisfied and define the $\Sb_{2n}$-valued matrices
\begin{align}\label{eq:large-dim}
 \Gamma:=\left( \begin{array}{cc}
 -RI_n&0 \\
  0& 0\end{array} \right), \quad
  \Psi(y):=\left( \begin{array}{cc}
 2RC^{-1}&0 \\
  0& \gamma\Sigma(y) \end{array} \right).
\end{align}
Then, there exists $\d_0>0$ such that for each $y\in \cD$ the matrix-valued Riccati equation 
\begin{align}\label{eq:matrixriccati}
-\Psi(y)-\Gamma  A(y)- A(y)\Gamma+  A(y) \left( \begin{array}{c}
 C \\
   I_n\end{array} \right) \Lambda^{-1}\left(C,  I_n \right)  A(y)=0
\end{align}
has a maximal solution $$ A(y)=\left( \begin{array}{cc}
 A_1(y)& A_{12}(y) \\
   A_{12}^\top(y)&  A_2 (y)\end{array} \right)\in \Sb_{2n}$$  
  for which the corresponding quadratic form
\begin{align}\label{def:hatvarpi}
\varpi(\xi,y) \mapsto\frac{1}{2}\xi^\top  A(y) \xi.
\end{align}
satisfies the following upper and lower bounds,
 \begin{align}\label{lower-bound-varpi}
 \sqrt{1+|y|^2}\frac{|\xi|^2}{\d_0}\geq \varpi(\xi,y)\geq \d_0|\xi|^2, \quad \mbox{for all $\xi\in\R^{2 n}$ and $y\in \cD$.}
 \end{align}
 \end{lem}

 \begin{proof}
See Appendix~\ref{sec:app}.
 \end{proof}

 The last assumption for our value expansion in Theorem~\ref{thm:expansion} is a comparison principle for a linear PDE.
 
\begin{asm}\label{assume-comparison}
Comparison holds for the \emph{second corrector equation}
\begin{align}\label{corrector2} 
\rho u(y) &=\cL^Y u(y)  +a(y).
\end{align}
among viscosity semisolutions $\phi$ satisfying,
\begin{align}\label{eq:comparison-class}
c(1+|\Sigma (y)|^2+M^*_\Sigma(y)+M^*_\cM(y) )\geq \phi(y) \geq  0, \quad \mbox{for some $c>0$.}
\end{align}
Here, the source term is:\footnote{The linear PDE~\eqref{corrector2} corresponds to the ``second corrector equation'' of Soner and Touzi~\cite{soner.touzi.13}; accordingly, the source term $a$ is the principal component of the value expansion \eqref{eq:exp}.}
 \begin{align}\label{eq:a}
 a(y):=\frac{1}{2}\mathrm{Tr}\left(c_\cM(y) A_2(y)\right), \quad \mbox{where } c_\cM(Y_t):=\frac{d\langle \cM(Y_t)\rangle}{dt} \in \Sb_{n}
 \end{align}
is the (infinitesimal) quadratic variation of the Merton portfolio. 
\end{asm}

In view of the positivity of $A_2$ (cf.~Lemma~\ref{lem:ricatti}), the following probabilistic representation of the function $u$ immediately shows that it is nonnegative.

\begin{lem}\label{lem:feynmann} 
Suppose Assumptions \ref{assume-expansion} and \ref{assume-comparison} hold and the function
\begin{align}\label{eq:fku}
u(y)=\E\left[\int_0^\infty e^{-\rho t}a(Y^y_t) dt\right]. 
\end{align}
satisfies \eqref{eq:comparison-class}. Then $u$ is the unique continuous viscosity solution of \eqref{corrector2} for $a$ defined in \eqref{eq:a}. 
 \end{lem}

 \begin{proof}
See Appendix~\ref{sec:app}.
 \end{proof}
 
 \begin{example}\label{ex:riccati}
 For the model with mean-reverting returns from Example~\ref{ex:OU}, Assumptions~\ref{assume-expansion} and \ref{assume-comparison} are satisfied, compare Appendix~\ref{s:app3}.
 \end{example}

 \subsection{Value Expansion}
 For all $\e>0$, denote by $\cJ^\e$ and $V^\e$ the mean-variance criterion and the value function corresponding to the asymptotic regime~\eqref{eq:scaling}. We are now ready to state our first main result, the large-liquidity expansion of the frictional value function.

 \begin{thm}\label{thm:expansion}
 Suppose Assumptions \ref{finite-V0}, \ref{assume-expansion}, and \ref{assume-comparison} are satisfied and the mapping $y\to A(y)$ from Lemma \ref{lem:ricatti} has locally bounded second-order derivatives. Then the frictional value function has the following expansion as $\e \to 0$, locally uniformly in $(d,h,y)\in \R^n\times\R^n\times\cD$,
 \begin{align}\label{eq:exp}
  &V^\e(\e d,h,y)\notag\\
   &=V^0(y)-\e\left(u(y)+h^\top d-\frac{d^\top C^{-1}d}{2}\right)-\e^2 \varpi\left(\frac{d}{\e^{1/2}},\frac{h-\cM(y)}{\e^{1/2}},y\right)+o(\e),\\
  &:=\hat{V}^\e(\e d,  h, y)+o(\e),\label{eq:defhatV}
 \end{align}
 with the functions $u$ and $\varpi$ from Lemmas~\ref{lem:feynmann} and~\ref{lem:ricatti}.
 \end{thm}
 
 \begin{proof}
 The main steps for the proof of this result are outlined in Section~\ref{ss.method} and subsequently carried out in the remainder of Section~\ref{s.estimates} as well as Sections~\ref{s.semilimits}, \ref{s.corrector.eq}, and \ref{s.first}.
 \end{proof}
 
The value expansion~\eqref{eq:exp} has two components, one stemming from dynamic trading over time and the other one from the initial conditions of the system. 

The ``dynamic component'' described by the function $u$ is similar to the corresponding expansions for models with only temporary trading costs. Indeed, the probabilistic representation \eqref{eq:fku} and \eqref{eq:a} show that the frictionless target strategy only enters through its (infinitesimal) quadratic variation here, just as for models with quadratic, proportional, or fixed costs, cf.~\cite{moreau.al.15} and the references therein. Whence, this ``portfolio gamma'' or ``activity rate'' is the crucial sensitivity of trading strategies with respect to small frictions also in the present setting where part of their effect only wears off gradually. The quadratic variation of the Merton portfolio is multiplied by the positive-definite matrix $A_2$ determined from the Riccati equation~\eqref{eq:matrixriccati}. If the resilience $R$ becomes large compared to the price impact parameters $\Lambda$ and $C$, one readily verifies that $A_2$ converges to the solution of the matrix equation $\gamma\Sigma=A_2 \Lambda^{-1} A_2$, that is 
\begin{equation}\label{eq:A2}
A_2=\Lambda^{1/2}(\Lambda^{-1/2}\gamma\Sigma\Lambda^{-1/2})^{1/2}\Lambda^{1/2}+o(1), \quad \mbox{as $R \to \infty$.}
\end{equation}
Whence, as resilience grows, temporary trading costs become the dominant friction and $A_2$ recovers the factor for purely temporary quadratic costs \cite[Remarks 4.5 and 4.6]{moreau.al.15}.\footnote{The same quantity also appears in a model with constant relative risk aversion, see \cite[Theorem 5]{guasoni.weber.13}.} The other comparative statics of this term are discussed in more detail in Section~\ref{ss:1d} for the one-dimensional case, where the Riccati equation~\eqref{eq:matrixriccati} can be solved explicitly.  

In addition to the dynamic component discussed so far, the value expansion~\eqref{eq:exp} also includes several terms that depend on the initial conditions. The quadratic form $\varpi$ is similar to its counterpart for purely temporary costs \cite[Theorem 4.3]{moreau.al.15}, in that it penalizes squared deviations of the initial portfolio from the frictionless optimum. Here, however, the initial distortion of the prices relative to their unaffected values also comes into play. In particular, the terms depending on the initial positions and displacements may have either a positive or a negative sign, unlike for purely temporary costs. The intuition is that a very large initial risky position may become favorable if the initial displacement is negative enough. Indeed, the mean reversion of the affected price to its unaffected value then leads to substantial extra positive returns, that may dominate the performance of the frictionless optimizer. However, these anomalies disappear if the initial price distortion is small enough.\footnote{The initial conditions would also disappear in the long-run limit if the discounted infinite-horizon criterion \eqref{eq:control-e} would be replaced by an ergodic goal functional as in \cite{dumas.luciano.91,gerhold.al.14,guasoni.mayerhofer.16}.}

 \subsection{Almost-Optimal Policy}
 
As our second main result, we now provide a family of ``almost-optimal'' policies $(\dot{H}^\e)_{\e>0}$ that achieves the leading-order optimal performance in the value expansion~\eqref{eq:exp}. To guarantee the admissibility of these policies, the following additional assumption is required.
 
   \begin{asm}\label{assume-admis}
There exists $\d_1>0$ such that the solution of the matrix Riccati equation \eqref{eq:matrixriccati} satisfies
$$\cL^YA(y)\leq (\rho-\d_1)A(y), \quad  \mbox{ for all $y\in \cD$}.$$
\end{asm}

\begin{example}
For the model with mean-reverting returns from Example~\ref{ex:OU}, the solution $A$ of \eqref{eq:matrixriccati} does not depend on the frictionless state variable $y$. Whence, Assumption~\ref{assume-admis} is satisfied for any strictly positive discount rate $\rho$ in this case. 
\end{example}

Using the solution to the Riccati equation \ref{eq:matrixriccati}, we can now identify the asymptotically optimal trading speeds that track the frictionless Merton portfolio and exploit the distortion of the asset prices relative to their unaffected values. 

 \begin{thm}\label{thm:policy} 
  Suppose the prerequisites of Theorem~\ref{thm:expansion} and Assumption~\ref{assume-admis} are satisfied and the function $u$ from Lemma \ref{lem:feynmann} is twice continuously differentiable. Define
   \begin{align}\label{eq:def-Q}
  Q_h(y):=(CA_{12}(y)+A_2(y))^\top, \qquad Q_d(y):=(CA_1(y)+A_{12}(y)^\top)^\top
  \end{align}
 and the feedback controls
\begin{align}\label{eq:asy-H}
\dot H^\e(\e d,h,y)=-{\Lambda^{-1} }\left(Q_d^\top(y)\frac{ d}{\e}+Q_h^\top(y) \frac{h-\cM(y)}{\e}\right).
\end{align}
Assume that for some $\d>0$,
 \begin{gather}\label{eq:assume:}
 \E\left[\int_0^\infty e^{-(\rho-\d)t }\Big(\mathrm{Tr}(A_2(Y_t) c_{\cM}(Y_t))+\sum_{i,j}\left|\tfrac{d\langle A^{i,j}, \cM^i\rangle_t}{dt}\right|^2\Big)dt\right]<\infty,\\
\label{eq:controlstate}\int_0^\infty e^{-\rho t}\E\left[\frac{1}{\e^2}| D_t^{\e d,h,y,\dot H^\e}|^2+|H_t^{\e d,h,y,\dot H^\e}-\cM(Y_t)|^2\right]dt=o(1) \quad \mbox{as $\e \to 0$,}\\
 \lim_{T\to \infty}\E\left[ e^{-\rho T}\hat V^\e(\e D^{\e d,h,y,\dot H^\e}_T,H_T^{\e d,h,y,\dot H^\e},Y_T)\right]=0, \quad \mbox{for all $\e>0$,\label{eq:cvhatV}}
\end{gather}
\mbox{and the local martingale in the It\^o decomposition of the approximate value function}\notag\\ 
\begin{gather} 
\hat V^\e(\e D^{\e d,h,y,\dot H^\e}, H^{\e d,h,y,\dot H^\e},Y) \label{eq:truemart} 
\end{gather}
 is a true martingale for all $\e>0$.
 Then the controls $\dot H^\e$ are admissible and asymptotically optimal in that, locally uniformly in $(d,h,y)$,
 $$\cJ^\e(d\e,h,y;\dot H^\e)=\hat{V}^\e(d\e,h,y)+o(\e)=V^\e(d\e,h,y)+o(\e), \quad \mbox{as $\e \to 0$.} $$
 \end{thm}
 
 \begin{proof}
See Section \ref{s.optimal.policy}.
 \end{proof}
 
  \begin{example}\label{ex:policy}
The regularity conditions of Theorem~\ref{thm:policy} are satisfied, in particular, for the model with mean-reverting returns from Example~\ref{ex:OU}. See Appendix~\ref{s:app3} for more details.
\end{example}

 As already observed in \cite{garleanu.pedersen.16}, the asymptotically optimal trading rates with temporary and transient price impact strike a balance between the following two objectives. On the one hand, they track the Merton portfolio, so as to remain near the optimal risk-return tradeoff in the frictionless model. On the other hand, they exploit the distortion of the asset prices relative to their unaffected values as an additional trading signal. 

The asymptotic formulas from Theorem~\ref{thm:policy} indentify the respective trading speeds through the matrix Riccati equation~\eqref{eq:matrixriccati}. As the resilience parameter $R$ becomes large, one readily verifies that $A_{12}$ becomes negligible, $A_1=C^{-1}+o(R^{-1})$, and $A_2$ is given by \eqref{eq:A2}. As a consequence, the asymptotically optimal (relative) trading speeds in this ``high-resilience regime'' are 
$$ 
\Lambda^{-1} Q_d^\top = \Lambda^{-1}+o(R^{-1}), \quad \mbox{and} \quad \Lambda^{-1}Q_h^\top=\Lambda^{-1/2}(\Lambda^{-1/2}\gamma\Sigma\Lambda^{-1/2})^{1/2}\Lambda^{1/2}+o(1),
$$
as $R \to \infty$. The second formula shows that as the resilience grows, we recover the asymptotically optimal trading rate for the model with purely temporary trading costs \cite[Theorem 4.7]{moreau.al.15}. In particular, this tracking speed only depends on the market, preference, and cost parameters, but not the optimal trading strategy at hand.

The corresponding coefficient of the price distortion has an even simpler form in the large-resilience limit. Indeed, it is independent of permanent impact, risk aversion, and price volatility. Instead, the exploitation of the displacement only trades off its size against the temporary trading cost.\footnote{Even though this coefficient does not vanish for $R \to \infty$, the effect of the distortion disappears in the high-resilience limit, because price distortions then disappear almost immediately.}
 
\subsection{Explicit Formulas for One Risky Asset}\label{ss:1d}

For a single risky asset ($n=1)$, the Riccati equation \eqref{eq:matrixriccati} can be solved explicitly. With the notation from \eqref{eq:def-Q}, we obtain
 \begin{align*}
 \frac{Q_d^2}{2R\Lambda}=-A_1 +\frac{1}{C}, \qquad  \frac{Q_hQ_d}{R\Lambda}=-A_{12}, \qquad \frac{Q_h^2}{\Lambda}=\gamma\Sigma.
 \end{align*}
As a consequence,
\begin{align*}
Q_h=\sqrt{\Lambda \gamma \Sigma} \qquad \mbox{and} \qquad \frac{Q_d^2}{2R\Lambda}+\frac{Q_d}{C}\left(\sqrt{\frac{\gamma\Sigma}{R^2\Lambda}}+1\right)-\frac{1}{C}=0.
\end{align*}
The last equation has one positive and one negative solution; the correct one is the positive one,
\begin{equation}
Q_d= \frac{\sqrt{\Lambda(\gamma\Sigma+R(2C+R\Lambda+2\sqrt{\gamma\Lambda\Sigma}))}-R\Lambda-\sqrt{\gamma\Sigma\Lambda}}{C}.
\end{equation}
 Indeed, the corresponding matrix
$$A=\left( \begin{array}{cc}
\frac{Q_d}{RC}\left(\sqrt{\frac{\gamma\Sigma}{\Lambda}}+R\right)& \frac{-Q_d}{R}\sqrt{\frac{ \gamma \Sigma}{\Lambda}} \\
   \frac{-Q_d}{R}\sqrt{\frac{ \gamma \Sigma}{\Lambda}} & \sqrt{\Lambda \gamma \Sigma}(1+\frac{CQ_d}{R\Lambda})\end{array} \right)$$
 then is positive as required for Lemma~\ref{lem:ricatti}, because both its trace and determinant are positive.
 
  Let us first discuss the lower-right entry $A_2=\sqrt{\Lambda \gamma \Sigma}(1+\frac{CQ_d}{R\Lambda})$ of this matrix, which multiplies the quadratic variation of the Merton portfolio in the leading-order term \eqref{eq:fku} of the value expansion~\eqref{eq:exp}. Its first summand is the corresponding term for the model with only temporary costs \cite[Formula (1.2)]{moreau.al.15}; whence, the second summand accounts for the additional effects of the transient price impact. Differentiation shows that this term is increasing in the permanent price impact parameter. For large resilience $R$, the first-order expansion is 
$$A_2=\sqrt{\gamma\Sigma\Lambda}\left(1+\frac{C/R}{\Lambda}\right)+O(R^{-2}).$$
This shows that the relative adjustment compared to the model with only temporary trading costs is substantial if the ratio $C/R$ of permanent price impact and resilence is large compared to the temporary price impact parameter $\Lambda$.
   
Next, note that the tracking speed for the frictionless Merton portfolio is always the same \emph{universal} quantity 
$$\frac{Q_h}{\Lambda}=\sqrt{\frac{\gamma\Sigma}{\Lambda}},$$
that already appears in the work of \cite{almgren.chriss.01} on optimal liquidation and is also asymptotically optimal in the model without transient price impact \cite{moreau.al.15}.\footnote{The same parameter also appears in more general liquidation problems \cite{schied.schoeneborn.09} and in linear-quadratic models with small information asymmetries, cf.~\cite{muhlekarbe.webster.16}.} Whence, this tracking speed is not only independent of the specific application for which the frictionless target strategy is designed, but is also the same with and without transient price impact. Note that for a single risky asset, this quantity obtains for any size of the transient impact parameters.

With transient price impact, the distortion of the price relative to its unaffected value is used as an additional trading signal. The corresponding weight is $Q_d/\Lambda$, which has the second-order expansion 
$$\frac{Q_d}{\Lambda} =  \frac{1}{\Lambda}-\frac{C+2\sqrt{\gamma\Sigma\Lambda}}{2\Lambda^2R}+O(R^{-2}).$$
The first term in this expansion is the universal high-resilience limit already identified in the multi-asset case above. The second-order term in turn shows it becomes more difficult to exploit price distortions if i) permanent price impact is large (so that initial distortions are offset quickly), ii) market risk is high relative to temporary trading costs (so that the Merton portfolio is tracked closely), or iii) resilience is low (so that the distortion only decays slowly and therefore can be exploited gradually).

\section{Outline of the Proof and Initial Estimates}\label{s.estimates}
\subsection{Outline of the Proof}\label{ss.method}
The proof of Theorem~\ref{thm:expansion} is based on stability results for viscosity solutions as in~\cite{soner.touzi.13,moreau.al.15}. However, due to the presence of the price distortion, these arguments cannot be applied directly to the value function at hand here. Instead, we first study a ``rescaled'' version $\tilde V^\e$ of the value function, defined in Section~\ref{ss.incorp}. We then establish an expansion for $\tilde V^\e$ and in turn use it to derive the expansion of the actual value function $V^\e$. 

To obtain the expansion of the rescaled value function, we first establish locally uniform bounds for $(\tilde V^\e-V^0)/\e$ in Section \ref{s.estimates}. This allows us to show in Section \ref{s.semilimits} that $u^\e$ -- the deviation of $\tilde V^\e$ from the frictionless value $V^0$ scaled with an appropriate power of $\e$ -- admits locally bounded upper and lower semilimits $u^*$ and $u_*$ which are upper and lower semicontinuous, respectively. 

In Section~\ref{s.corrector.eq} we then establish that $y\to u^*(0,\cM(y),y)$ and $y\to u_*(0,\cM(y),y)$ (i.e, the semilimits evaluated along the frictionless versions of their state variables) are viscosity sub- and supersolutions, respectively, of the second corrector equation~\eqref{corrector2}. Together with our estimates on $u^\e$, the comparison principle for \eqref{corrector2} from Assumption~\ref{assume-comparison} in turn yields $u^*(0,\cM(y),y)=u_*(0,\cM(y),y)$ for all $y\in \cD$.

Finally, in Section~\ref{s.first}, we use the viscosity properties of $u^*$ and $u_*$ to identify the dependence of these semilimits on the frictional state variables $(d,h)$. As the same result obtains for both functions, it follows that the shifted value function indeed has the expansion posited in Theorem~\ref{thm:expansion}.

\subsection{The Rescaled Value Function}\label{ss.incorp}

As liquidity becomes large in our critical regime \eqref{eq:scaling}, we expect that the price distortion $D$ tends to zero. To charaterize its limit behavior, $D$ therefore needs to be rescaled appropriately. In Lemma \ref{lem:frictional-expansion}, we show that the natural scaling of $D$ is of order $\e$. In order to expand the value function for small $\e$, it is therefore natural to study its rescaling $V^\e(\e d,h,y)$. However, it turns out that the asymptotic analysis of this function is severely complicated by the fact that it is \emph{not} uniformly bounded from above by the frictionless value function for all arguments, e.g., if the agent starts with a large positive risky position and the initial price distortion is sufficiently negative. As a way out, we study the asymptotic expansion of the following shifted version of $V^\e(\e d,h,y)$  instead,
\begin{align}\label{dfn:tildeV}
\tilde V^\e(d,h,y)=V^\e(d\e,h,y)+\e\left(h^\top d-\frac{d^\top   C^{-1}d}{2}\right).
\end{align}
In Lemma \ref{lem:frictional-expansion}, we are then able to show that $\tilde V^\e$ is bounded from above by the frictionless value function for small $\e>0$. After deriving the limiting results for this function, we can then in turn deduce the expansion of the actual value function and a corresponding asymptotically optimal policy. 

The viscosity property of the value function (cf.~Proposition~\ref{asm:frictional.pde}) and a direct calculation show that the rescaled value function $\tilde V^\e$ is a (possibly discontinuous, compare \cite[Definition 4.2]{fleming.soner.06}) viscosity solution of  
\begin{align}\label{eq:pdetilde}
\rho \tilde V^\e +G^\e(\cdot,{\pa_d \tilde V^\e},{\pa_h \tilde V^\e},\pa_y \tilde V^\e,\pa_{yy}\tilde V^\e)=0,
\end{align}
where
\begin{align}\label{dfn:G}
&G^\e(d,h,y,p_1,p_2,p_3,X):=\\
&\qquad-\mu_Y^\top  (y)p_3-\frac{1}{2}\mathrm{Tr}\left(\sigma_Y(y)\sigma_Y^\top  (y)X\right)-\frac{\mu^\top(y)(\gamma\Sigma(y))^{-1}\mu(y)}{2}\\
&\qquad+f(d,h-\cM(y),y)-\cH\left(d,h-\cM(y),y,\frac{p_1}{\e},\frac{p_2}{\e}\right)-\e\rho \left(h^\top d-\frac{d^\top   C^{-1}d}{2}\right). \notag
\end{align}
Here, for all $(\xi,y)=(\xi_1,\xi_2,y)\in R^{2n}\times \cD$ and $p=(p_1,p_2)\in \R^{2n}$, the source term and the convex Hamiltonian are
\begin{align}\label{eq:source}
 f(\xi,y)&:=\frac{\gamma}{2}(\xi_2)^\top \Sigma(y) \xi_2+R\xi_1^\top C^{-1}\xi_1\\
 \cH(\xi,y,p_1,p_2)&:=-R\xi_1^\top p_1+\frac{1}{2}(C p_1+p_2)^\top  \Lambda^{-1}(C p_1+p_2)\notag\\
 &=-R\xi_1^\top p_1+\frac{1}{2}p^\top  \hat C \Lambda^{-1}\hat C^\top p \label{eq:source2}
\end{align}
with $\hat{C}=(C, Id)^\top$.

 \begin{rem}\label{rem:ham}
Note that as a symmetric matrix of dimension $2n$, $\hat C \Lambda^{-1}\hat C^\top$ is degenerate. One of the main technical challenges of this paper is the consequence that the Hamiltonian $\cH$ is degenerate in $p$. This is not the case in \cite{moreau.al.15} where the problem remains $n$ dimensional and the non-degeneracy of $\Lambda$ is a sufficient assumption to establish the viscosity property of the semilimits defined below. Here, we show  that the non-degeneracy assumption of the Hamiltonian $\cH$ can be replaced by the existence of positive solutions for a matrix-valued Riccati equation, cf.~\eqref{eq:matrixriccati}.
\end{rem}

\subsection{A Uniform bound}

In this section, we derive some elementary moment estimates for Ornstein-Uhlenbeck-type processes. These will be used to derive bounds for the semilimits introduced in Section \ref{s.semilimits} below.

\begin{lem}\label{expansion}
Let $\tilde \phi,\phi^i:\cD\mapsto\R$, $i=1,2$ be twice continuously differentiable functions. Define the semimartingales
$$M^i_t=\phi^i(Y_t) \quad \mbox{and} \quad \tilde M_t=\tilde \phi(Y_t),$$ 
and the corresponding Ornstein-Uhlenbeck-type processes
 $$dX_t^i=-\frac{\lambda_i}{\e}X_t^i dt +dM^i_t, \quad \lambda_i>0 \mbox{ for } i=1,2,$$
and set $\lambda=\lambda_1\wedge\lambda_2\wedge 2>0$. Then, there exists a (sufficiently large) constant $c>0$ such that for all $t>0$ and $\e\in(0,\frac{\lambda}{2})$,
\begin{align}\label{cont-1term}
\left|\E\left[X_t^i \tilde M_t\right]\right|&\leq c e^{-\lambda_i t/\e}|X_0^i \tilde M_0|+\frac{c\e}{\lambda_i}(X_0^i)^{2}\\
&+\frac{c\e}{\lambda_i}\sup_{0\leq s\leq t}\E\left[|\tilde M_s |^2+|\cL^Y\tilde \phi_s|^{2}+|(\sigma_Y^\top \pa_y \tilde\phi)_s|^{2}+|\cL^Y\phi^i_s|^{2}+|(\sigma_Y^\top \pa_y \phi^i)_s|^{2}\right]
\notag
\end{align}
and 
\begin{align}
\left|\E\left[X_t^iX_t^j \tilde M_t\right]\right| &\leq c e^{-\lambda t/\e}\left(|X_0^iX_0^j \tilde M_0|+2 +|X_0^i|^4+|X_0^j|^4\right)\label{control-quadratics}\\
&+\frac{c\e}{\lambda}\sup_{0\leq s\leq t}\E\Big[1+|\tilde M_s|^4+|\cL^Y \tilde \phi_s|^2+ |\cL^Y \phi^j_s|^4+|\cL^Y\phi^i_s|^{4}\notag\\
&\qquad \qquad \qquad+|(\sigma_Y^\top \pa_y \phi^i)_s|^{4}+|(\sigma_Y^\top \pa_y \phi^j)_s|^{4}+|(\sigma_Y^\top \pa_y \tilde \phi)_s|^{4} \Big].\notag
\end{align}
\end{lem}

\begin{proof}
{\it Step 1:} It\^o's formula applied to $(X_t^i)^{2k}$ and the $\epsilon$-Young inequality show that there exists $c>0$ such that 
$$
\E[(X_t^i)^{2k}]\leq c\left((X_0^i)^{2k}+\sup_{0\leq s\leq t}\E\left[|(\cL^Y\phi^i)_s|^{2k}+|(\sigma_Y^\top \pa_y \phi^i)_s|^{2k}\right]\right),$$
for $k=1,2$ and $\e\in(0,\frac{\lambda}{2})$.

{\it Step 2: } Apply It\^o's formula to $X_t^i \tilde M_t$ and solve the ODE for $\E[X_t^i \tilde M_t]$, obtaining 
\begin{align*}
\E[X_t^i \tilde M_t] = &e^{-\lambda_i t/\e}X_0^i \tilde M_0\\
&\quad+\int_0^t \e^{-\lambda_i(t-s)/\e}\E\left[\tilde M_s \cL^Y\phi^i_s+ (\sigma_Y^\top\pa_y \phi^i)_s(\sigma_Y^\top\pa_y \tilde \phi)_s+X_s^i \cL^Y \tilde \phi_s\right]ds.
\end{align*}
Together with the inequality from Step 1, it follows that, for some $c>0$,
\begin{align*}
\left|\E\left[X_t^i \tilde M_t\right]\right|&\leq c e^{-\lambda_i t/\e}|X_0^i \tilde M_0|+\frac{c\e}{\lambda_i}(X_0^i)^{2}\\
&+\frac{c\e}{\lambda_i}\sup_{0\leq s\leq t}\E\left[|\tilde M_s |^2+|\cL^Y\tilde \phi_s|^{2}+|(\sigma_Y^\top \pa_y \tilde\phi)_s|^{2}+|\cL^Y\phi^i_s|^{2}+|(\sigma_Y^\top \pa_y \phi^i)_s|^{2}\right].
\end{align*}
This establishes the first part of the assertion.

{\it Step 3: } As in Step 2, It\^o's formula applied to $X_t^iX_t^j \tilde M_t$ gives
\begin{align*}
&\E\left[X_t^iX_t^j \tilde M_t\right]\\
&=e^{-(\lambda_i+\lambda_j)t/\e}X_0^iX_0^j \tilde M_0\\
&\quad+\int_0^t \e^{-(\lambda_i+\lambda_j)(t-s)/\e} \E\left[\tilde M_s \left((\sigma_Y^\top\pa_y \phi^j)_s(\sigma_Y^\top\pa_y  \phi^i)_s+ X_s^i \cL^Y \phi^j_s+X_s^j\cL^Y \phi^i_s\right)\right] ds\\
&\quad+\int_0^t \e^{-(\lambda_i+\lambda_j)(t-s)/\e} \E\left[(\sigma_Y^\top\pa_y\tilde \phi)_s (X_s^i (\sigma_Y^\top\pa_y \phi^j)_s+X_s^j(\sigma_Y^\top\pa_y \phi^i)_s)+ X_s^i X_s^j\cL^Y \tilde \phi_s\right] ds.
\end{align*}
Combined with the estimate from Step 1, this identity shows that, for some $c>0$, 
\begin{align*}
&\left|\E\left[X_t^iX_t^j \tilde M_t\right]\right|\leq e^{-(\lambda_i+\lambda_j)t/\e}|X_0^iX_0^j \tilde M_0|\\
&+c\int_0^t e^{-(\lambda_i+\lambda_j)(t-s)/\e} \E\left[|\tilde M_s|^2 + |(\sigma_Y^\top\pa_y \phi^j)_s(\sigma_Y^\top\pa_y  \phi^i)_s|^2+ |X_s^i|^2\right.\\
&\left.\quad\quad\quad\quad\quad\quad\quad\quad\quad\quad+ |\tilde M_s\cL^Y \phi^j_s|^2+|X_s^j|^2+|\tilde M_s \cL^Y \phi^i_s|^2 \right] ds\\
&+c\int_0^t e^{-(\lambda_i+\lambda_j)(t-s)/\e} \E\left[|(\sigma_Y^\top\pa_y\tilde \phi)_s  (\sigma_Y^\top\pa_y \phi^j)_s|^2+|(\sigma_Y^\top\pa_y\tilde \phi)_s(\sigma_Y^\top\pa_y \phi^i)_s|^2\right.\\
&\left.\quad\quad\quad\quad\quad\quad\quad\quad\quad\quad+ |X_s^i |^4+|X_s^j|^4+|\cL^Y \tilde \phi_s|^2\right] ds\\
&\leq e^{-\lambda t/\e}\left(|X_0^iX_0^j \tilde M_0|+2 +|X_0^i|^4+|X_0^j|^4\right)+\frac{c\e}{\lambda}\sup_{0\leq s\leq t}\E\left[1+|\tilde M_s|^4+|\cL^Y \tilde \phi_s|^2 \right.\\
&\quad\left.+ |\cL^Y \phi^j_s|^4+| \cL^Y \phi^i_s|^4+|\cL^Y\phi^i_s|^{4}+|(\sigma_Y^\top \pa_y \phi^i)_s|^{4}+|(\sigma_Y^\top \pa_y \phi^j)_s|^{4}+|(\sigma_Y^\top \pa_y \tilde \phi)_s|^{4} \right],
\end{align*}
as claimed. 
\end{proof}

\subsection{Expansion Along a Class of Policies}
\label{ss.ansatz}
Note that the value function $V^\e(\e d,h,y)$ uses the state variable $D^{(\e d,h,y),\dot H_t}$ which, for $\dot H\in \cA_\rho^\e$, satisfies
$$D^{(\e d,h,y),\dot H}_0=\e d \quad \mbox{ and } \quad dD^{(\e d,h,y),\dot H}_t= \left(-R{\e^{-1}}D^{(\e d,h,y),\dot H}_t + C\e\dot H_t\right) dt.$$
To simplify notation, we pass to $\tilde D^{(\e d,h,y),\dot H_t}=\e^{-1}D^{(\e d,h,y),\dot H_t}$, which satisfies
\begin{align}\label{def:tildeD}
\tilde D^{(\e d,h,y),\dot H}_0=d \quad \mbox{ and } \quad d\tilde D^{(\e d,h,y),\dot H}_t= \left(-R{\e^{-1}}\tilde D^{(\e d,h,y),\dot H_t}_t + C\dot H_t\right) dt.
\end{align}

We now apply Lemma~\ref{lem:conc} in the present large-liquidity context to derive the following uniform upper bound, valid for any admissible strategy.

\begin{lem}\label{lem:frictional-expansion}
Suppose Assumptions \ref{finite-V0} and \ref{assume-expansion} are satisfied. Then there exists $\e_0>0$ such that, for all $\e\in(0,\e_0)$ and for all $\dot H\in \cA_{\rho}^\e$, $\theta=(\e d,h,y)$,
\begin{align*}
\notag & \frac{\cJ^\e(\e d,h,y; \dot{H})-V^0(y)}{\e}=-h^\top d+\frac{d^\top   C^{-1}d}{2}+\rho\E\left[\int_0^\infty e^{-\rho t}(H_t^{\theta,\dot H} )^\top  \tilde D_t^{\theta,\dot H} dt\right]\notag\\
&\quad\quad\quad\quad-\E\left[\int_0^\infty e^{-\rho t}(\tilde D_t^{\theta,\dot H})^\top\left(\frac{(2R\e^{-1}+\rho) C^{-1} }{2}\right)\tilde D_t^{\theta,\dot H}dt\right]\\
&\quad\quad\quad\quad-\E\left[\int_0^\infty \frac{e^{-\rho t}}{2}\left((H_t^{\theta,\dot H}-\cM_t)^\top\e^{-1}\gamma\Sigma_t (H_t^{\theta,\dot H}-\cM_t)+\e\dot{H}_t^\top\Lambda_t\dot{H}_t \right)dt\right].\notag
\end{align*}
\end{lem}

\begin{proof}
It suffices to verify that the prerequisites of Lemma \ref{lem:conc} are satisfied. Whence, we need to show that the family 
$$\left\{ \left( {\begin{array}{cc}
   -(2R\e^{-1}+\rho)C^{-1} & \rho I_n \\       \rho I_n & -\e^{-1}\gamma\Sigma(y) \      \end{array} }\right)
:y\in \cD\right\}$$
is bounded from below by a symmetric negative matrix for sufficiently small $\e>0$. Let $M$ be a matrix in this family and $\xi=(\xi_1^\top,\xi_2^\top)^\top\in \R^{2n}$. Then,
\begin{align*}\xi^\top M\xi&=-\e^{-1}\xi_1^\top(2R\e^{-1}+\rho)C^{-1}\xi+2\rho \xi_1^\top \xi_2-\gamma\xi_2^\top \Sigma(y)\xi_2\\
&\leq -2 \e^{-2}\xi_1^\top RC^{-1}\xi_1-\gamma\xi_2^\top \Sigma(y)\xi_2+2\rho \xi_1^\top \xi_2\\
&\leq -\e^{-2}\xi_1^\top RC^{-1}\xi_1-\xi_2^\top(\gamma \Sigma(y)-\e^{2}\rho^2R^{-1}C)\xi_2.
\end{align*}
Note that $\rho^2R^{-1}C$ is constant and $\gamma \Sigma(y)\geq \underline m I_n$ by Assumption~\ref{assume-expansion}(i). Hence there exists $\e_0>0$ such that for all $\e \in(0,\e_0)$ we have $\gamma \Sigma(y)-\e^{2}\rho^2R^{-1}C\geq \gamma \Sigma(y)/2$ and 
$$ -\e^{-2}\xi_1^\top RC^{-1}\xi_1-\xi_2^\top(\gamma \Sigma(y)-\e^{2}\rho^2R^{-1}C)\xi_2  \leq -\underline m |\xi|^2,$$
establishing the required uniform lower bound for the above family.
\end{proof}

To control the semilimits studied in Section~\ref{s.semilimits} below, we define the following class of suboptimal but simple feedback controls,
\begin{align}\label{eq:constcontrol}
\dot H^{\e,\a I_n,\a C}(\e d,h,y)=-\frac{\a}{\e} (h-\cM(y))-\frac{\a C}{\e}{d}.
\end{align}

For this parametric class, we have the following estimate.

\begin{prop}\label{prop:expansion}
Suppose Assumption~\ref{assume-expansion} is satisfied and fix $\a>0$. Then there exists constants $c, \e_\a>0$ such that, for all $\e \in(0,\e_\a)$ and $(d,h,y)\in\R^n\times\R^n\times \cD$,
\begin{align}\label{eq:expansion}
\notag&\left|\frac{\cJ^\e(d\e,h,y;\dot H^{\e,\a I_n,\a C})-V^0(y)}{\e} +h^\top d-\frac{d^\top   C^{-1}d}{2}\right|\\
&\qquad\leq c \left(1+|d|^4+|h-\cM(y)|^4+|\Sigma (y)|^2+M_\Sigma(y)+M_\cM(y)\right).
\end{align}
\end{prop} 

\begin{proof}
Fix a control $\dot H^{\e,\a I_n,\a C}$ as in \eqref{eq:constcontrol} and set 
$$
N=\left( \begin{array}{cc}
RI_n+\a C^2 &  \a C  \\
\a C & \a I_n   \end{array} \right).
$$
Then, in matrix-vector notation, the corresponding state dynamics are
 $$dX_t:=d\left( \begin{array}{c}
\tilde D_t \\
 H_t-\cM_t \end{array} \right)=-\frac{N}{\e} X_t dt-d\bar\cM_t, \quad \mbox{where } \bar \cM:=\left( \begin{array}{c}
0 \\
 \cM\end{array} \right).$$

In view of \cite[Theorem 3]{silvester2000determinants}, the matrix $N$ is symmetric positive definite. Thus, it is diagonalisable and can be written as 
$$N=BMB^{-1},$$
where $M=\mathrm{diag}[\mu_1,\ldots,\mu_{2n}]$ is the diagonal matrix with entries $\mu_i>0$ and $\mu_1=\min_{i \in \{1,\ldots,2n\}} \mu_i$, $\mu_{2n}=\max_{i \in \{1,\ldots,2n\}} \mu_i$  
The rescaled state variable 
$$\tilde X_t=B^{-1}X_t$$ 
in turn has dynamics $d\tilde X_t=-\e^{-1}M\tilde X_t dt-d(B^{-1}\bar\cM_t)$ or, equivalently,
$$d\tilde X_t^i:=-\frac{\mu_i}{\e} \tilde X_t^i dt-d(B^{-1}\bar\cM)^i_t, \quad i=1\ldots 2n.$$
Whence, after these transformations the components of $\tilde X_t$ satisfy the assumptions of Lemma~\ref{expansion}, so that the estimates provided there can be used to bound moments involving $\tilde{X}_t$. To bring this to bear, we now use Lemma~\ref{lem:frictional-expansion} to express the quantity we want to bound in terms of $\tilde X_t$ as follows,
\begin{align}\label{eq:withtildeX}
&\frac{\cJ^\e(d\e,h,y;\dot H^{\e,\a I_n,\a C})-V^0(y)}{\e} +h^\top d-\frac{d^\top   C^{-1}d}{2}=\rho\E\left[\int_0^\infty e^{-\rho t}\cM_t^\top \tilde D_t dt\right]\notag\\
&\quad\quad-\frac{1}{2}\E\left[\int_0^\infty e^{-\rho t}X_t^\top\left(\left( \begin{array}{cc}
{(2R\e^{-1}+\rho) C^{-1} } &  -\rho I_n \notag \\
-\rho I_n   & \e^{-1}\gamma\Sigma_t   \end{array} \right)+\a^2\hat C\Lambda^{-1}\hat C^\top \right)X_tdt\right]\\
&=\rho\E\left[\int_0^\infty e^{-\rho t}(\cM_t^\top,0)B\tilde X_t dt\right]-\frac{1}{2}\E\left[\int_0^\infty e^{-\rho t}\tilde X_t^\top B^\top\left(\begin{array}{cc}
{\rho C^{-1} } &  -\rho I_n  \\
-\rho I_n   &0  \end{array} \right) B\tilde X_tdt\right]\notag\\
&\quad\quad-\frac{\e^{-1}}{2}\E\left[\int_0^\infty e^{-\rho t}\tilde X_t^\top B^\top\left(\left( \begin{array}{cc}
{2RC^{-1}} & 0  \\
0   & \gamma\Sigma_t   \end{array} \right)+\a^2\hat C\Lambda^{-1}\hat C^\top \right)B\tilde X_tdt\right]\notag\\
\end{align}
We first use \eqref{cont-1term} to obtain the following bound (for a generic constant $c>0$ depending on $\mu_1, \mu_{2n}$),
\begin{align*}
&\rho\left|\E\left[\int_0^\infty e^{-\rho t}(\cM_t^\top,0)B\tilde X_t dt\right]\right|\leq c \e \left(|d^2|+|h-\cM(y)|^2\right)\\
&\quad\quad+ c \e \int_0^\infty e^{-\rho t} \sup_{0\leq s \leq t}\E\left[|\cM(Y_s)|^2+|\cL^Y\cM(Y_s)|^2+|\sigma_Y^\top \pa_y \cM(Y_s)|^2\right]dt\\
&\leq c \e \left(1+|d|^2+|h-\cM(y)|^2+M_\cM(y)\right).
\end{align*}
Similarly, using \eqref{control-quadratics}, we can control the absolute value of the last two terms of the right-hand side of \eqref{eq:withtildeX} as follows ($c>0$ is again a generic constant depending on $\mu_1, \mu_{2n}$):
\begin{align*}
 & \left(1+|d|^4+|h-\cM(y)|^4+|\Sigma (y)|^2\right)+\int_0^\infty e^{-\rho t}\sup_{0\leq s\leq t}\E\left[1+|\Sigma_s|^4+|\cL^Y \Sigma_s|^2 \right.\\
 &\qquad\qquad\qquad\qquad\qquad\left.+ |\cL^Y\cM_s|^{4}+|(\sigma_Y^\top \pa_y \cM)_s|^{4}+|(\sigma_Y^\top \pa_y \Sigma)_s|^{4} \right]dt.\\
 &\qquad \leq c \left(1+|d|^4+|h-\cM(y)|^4+|\Sigma (y)|^2+M_\Sigma(y)+M_\cM(y)\right).
\end{align*}
Together, these two estimates show that there exists a constant $c>0$ such that, for all sufficiently small $\e>0$:  
\begin{align*}
\notag&\left|\frac{\cJ^\e(d\e,h,y;\dot H^{\e,\a I_n,\a C})-V^0(y)}{\e} +h^\top d-\frac{d^\top   C^{-1}d}{2}\right|\\
&\qquad\leq c \left(1+|d|^4+|h-\cM(y)|^4+|\Sigma (y)|^2+M_\Sigma(y)+M_\cM(y)\right).
\end{align*}
Whence, the bound \eqref{eq:expansion} is indeed satisfied.
\end{proof}

 \section{Semilimits}\label{s.semilimits}
 As discussed in Section~\ref{ss.incorp}, to establish the value expansion in Theorem~\ref{thm:expansion} we study the rescaled function 
 \begin{equation}\label{eq:def-ue}
  u^\e(d,h,y):=\frac{V^0(y)-\tilde V^\e(d,h,y)}{\e}=\frac{V^0(y)- V^\e(\e d,h,y)}{\e}-h^\top d+\frac{d^\top   C^{-1}d}{2}.
  \end{equation}
 As this object has no a-priori regularity, we follow \cite{soner.touzi.13,moreau.al.15} and consider its upper and lower semicontinuous envelopes,
\begin{equation*}
 u^{*,\e}(d,h,y):=\frac{V^0(y)-\tilde V^{\e}_*(d,h,y)}{\e},\qquad u^\e_*(d,h,y):=\frac{V^0(y)-\tilde V^{*,\e}(d,h,y)}{\e},
 \end{equation*}
and the corresponding upper and lower ``semilimits'',
\begin{align}\label{eq:semi-limits-bar}
 u^*(\theta):=\limsup_{\substack{\e\to 0,\\ \theta'\to \theta }} u^{*,\e}(\theta'), \qquad u_*(\theta):=\liminf_{\substack{\e\to 0,\\ \theta'\to \theta }} u^\e_*(\theta'), \qquad \theta=(d,h,y).
\end{align}
By definition, $ u^*(\theta)\geq u_*(\theta)$. In the following sections, we use viscosity techniques to establish the converse inequality. Then, the two semilimits coincide and -- again by definition -- also equal the actual limit \eqref{eq:def-ue} we are interested in. 

The first step to carry out this program is to show that the ratio \eqref{eq:def-ue} is locally bounded, so that its upper and lower envelopes are indeed finite, upper and lower semicontinuous functions.

\begin{prop}\label{local-bound-ratios}
Suppose Assumption~\ref{assume-expansion} is satisfied. Then there exists $\e_0>0$ such that $(d,h,y) \mapsto \sup_{\e\in(0,\e_0)}|u^\e(d,h,y)|$
is locally bounded on $\R^n\times\R^n\times \cD$. Hence $u^*$ and $u_*$ are upper and lower semicontinuous functions, respectively, and there exists a constant $c>0$ such that 
\begin{align*}c \left(1+|d|^4+|h-\cM(y)|^4+|\Sigma (y)|^2+M^*_\Sigma(y)+M_\cM^*(y)\right)& \geq u^*(d,h,y)\\
&\geq u_* (d,h,y)\geq  0.
\end{align*}
\end{prop}
\begin{proof}{\it Step 1: lower bound.} By Lemma \ref{lem:frictional-expansion}, for $\theta=(\e d, h, y)$,
\begin{align*}
&\frac{V^0(y)-\cJ^\e(\e d,h,y; \dot{H})}{\e}-h^\top d+\frac{d^\top   C^{-1}d}{2}\\
&=-\rho\E\left[\int_0^\infty e^{-\rho t}(H_t^{\theta,\dot H} )^\top \tilde D_t^{\theta,\dot H} dt\right]+\E\left[\int_0^\infty e^{-\rho t}(\tilde D_t^{\theta,\dot H})^\top\frac{(2R\e^{-1}+\rho) C^{-1} }{2}\tilde D_t^{\theta,\dot H}dt\right]\\
&\quad+\E\left[\int_0^\infty \frac{e^{-\rho r}}{2\e}\left((H_t^{\theta,\dot H}-\cM_t)^\top\gamma\Sigma_t (H_t^{\theta,\dot H}-\cM_t)+\e\dot{H}_t^\top\Lambda_t\dot{H}_t \right)dt\right].\notag
\end{align*}
As shown in the proof of Lemma~\ref{lem:frictional-expansion} the following matrix is positive for sufficiently small $\e$,
 $$\left( {\begin{array}{cc}
  {(2R\e^{-1}+\rho) C^{-1}} & -\rho I_n \\      - \rho I_n & \e^{-1}\gamma\Sigma(y) \      \end{array} } \right).$$
  Whence, the following lower bound is valid for all sufficiently small $\e>0$,
$$\frac{1}{\e}(V^0(y)-\cJ^\e(\e d,h,y,\dot H))- h^\top d+\frac{d^\top   C^{-1}d}{2}\geq-\rho\E\left[\int_0^\infty e^{-\rho t}(\cM_t)^\top \tilde D_t^{\theta,\dot H} dt\right].$$

By Lemma \ref{expansion} the following convergence holds locally uniformly in $(d,h,y)$ as $\e\to 0$, which shows that \eqref{eq:def-ue} is indeed bounded from below:
$\E\left[\int_0^\infty e^{-\rho t}(\cM_t)^\top \tilde D_t^{\theta,\dot H} dt\right]\to 0.$

{\it Step 2: upper bound.} We now derive an upper bound for $u^\e$. By definition of $V^\e$ it is sufficient to find a family of control $\dot H^\e\in \cA^\e_\rho$ and $\e_0>0$ such that  $(V^0(y)-\cJ^\e(d\e,h,y;\dot H^\e))/\e$
is bounded from above by an appropriate function for all $\e\in(0,\e_0)$. 

For fixed $\a>0$, we use the strategy $\dot H^{\e,\a I_n,\a C}$ from \eqref{eq:constcontrol}. Note that, as mentioned in the proof of Proposition \ref{prop:expansion},
the matrix 
$$I_{2n}\left( \begin{array}{cc}
RI_n+\a C^2 &  \a C  \\
\a C & \a I_n   \end{array} \right)+I_{2n} \left( \begin{array}{cc}
RI_n+\a C^2 &  \a C  \\
\a C & \a I_n   \end{array} \right)=2\left( \begin{array}{cc}
RI_n+\a C^2 &  \a C  \\
\a C & \a I_n   \end{array} \right)$$
has only positive eigenvalues. Thus, Proposition \ref{prop:addmisibility} shows that these trading rates are admissible. For sufficiently small $\e$, Proposition~\ref{prop:expansion} yields
\begin{align*}
&\frac{V^0(y)-\cJ^\e(d\e,h,y;\dot H^\e)}{\e}   -h^\top d+\frac{d^\top   C^{-1}d}{2}\\
&\quad\leq  c \left(1+|d|^4+|h-\cM(y)|^4+|\Sigma (y)|^2+M_\Sigma(y)+M_\cM(y)\right),
\end{align*}
which in turn gives the result. 
\end{proof}

\section{Corrector Equations}\label{s.corrector.eq}
In this and the subsequent section, we show that for all $(d,h,y)\in \R^n\times\R^n\times \cD$ we have 
\begin{equation}\label{eq:bla}
u^*(d,h,y)=u_*(d,h,y)=u(y)+\varpi(d,h-\cM(y),y),
\end{equation}
where $u(y)$ is the solution of the second corrector equation from Lemma~\ref{lem:feynmann}. As a consequence, the limit of \eqref{eq:def-ue} indeed exists and is given by this expression in line with the expansion from Theorem~\ref{thm:expansion}. To this end, we first establish that the semilimits, evaluated along the frictionless state variables,
\begin{align}\label{eq:barup}
   u^* (y):=   u^*(0,\cM(y),y)\\
   u_* (y):=   u_*(0,\cM(y),y)\label{eq:bardown}
\end{align}
are viscosity sub- and supersolutions, respectively, of the second corrector equation~\eqref{corrector2}. The comparison principle for the second corrector equation (cf.~Assumption~\ref{assume-comparison}) in turn implies 
$$u^* (y)=  u_* (y).$$

We then show in Section \ref{s.first} that, as a function of $(d,h)$, $u^*(d,h,y)$ and $u_*(d,h,y)$ are viscosity sub- and supersolutions of a first-order PDE. The equality $u^*(0,\cM(y),y)=u_*(0,\cM(y),y)$
and a comparison result for this first-order equation in turn imply that the dependence of both functions in $(d,h)$ is given by \eqref{eq:bla}.

\subsection{Notations for the Proof of the Viscosity Property}
Set
\begin{align}\label{def:E1}
E_1(\xi_1,\xi_2,y,p_1,p_2)&:=\frac{\gamma \xi_2^\top\Sigma\xi_2}{2}+R\xi_1^\top C^{-1}\xi_1-Rp_1^\top\xi_1-\frac{1}{2}(Cp_1 +p_2)^\top  \Lambda^{-1} (Cp_1 +p_2)\\
&= f(\xi,y)- \cH(\xi,y,-p_1,-p_2).\notag
\end{align}
and  define the following differential operator acting on smooth functions,
\begin{align}\label{def:quadM}
\pa_{\cM\cM}w(\xi_1,\xi_2,y)=\mathrm{Tr}\left(c_\cM(y) \pa_{2,2}w(\xi_1,\xi_2,y)\right)
\end{align}
where $\pa_{2,2}w $ is the $n\times n$ dimensional Hessian matrix of $w$ corresponding to $\xi_2\in \R^n$.
Similarly to \cite[Definition 3.6]{moreau.al.15} and \cite[Definition 3.1]{soner.touzi.13}, we now define the ``first corrector equation'' for our expansion.

\begin{dfn}
Fix $y\in \cD$. The \emph{first corrector equation} is the differential equation 
\begin{align}\label{eq:fst}
-\frac{1}{2}\pa_{\cM\cM}w(\xi,y)-E_1(\xi,y,w(\xi,y),\pa_\xi w(\xi,y))+\tilde a(y)=0,\mbox{ for all }\xi\in\R^{2n},
\end{align}
where the unknown quantity is the couple $w:\R^{2n}\times \cD\to \R$ and $\tilde a:\cD\to \R$.
\end{dfn}
A direct computation shows that
$$E_1(\xi,y,\varpi(\xi,y),\pa_\xi\varpi(\xi,y))=0, \quad \mbox{for all $(\xi,y)\in\R^{2n}\times \cD$}$$
for the quadratic form $\varpi$ from \eqref{def:hatvarpi}. By definition of $a$ in \eqref{eq:a}, it in turn follows that the couple $(\varpi,a)$ is a solution of the first corrector equation~\eqref{eq:fst}.

Also note that due to our regularity assumptions on $\mu_Y$ and $\sigma_Y$, the following function is locally bounded on $\cD$,
\begin{align}\label{eq:barM}
\bar M(y):=\max\{1,|\mu_Y|,|\sigma_Y|,|\cM|,|\pa_y \cM|,|\pa_{yy}\cM|\}(y).
\end{align}

 \subsection{Expansion of the Generator}
 We now expand the generator of the PDE \eqref{eq:pdetilde}. To simplify notation, we set 
\begin{align}
 ^\e \xi:=\left( \begin{array}{c}
^\e \xi_1 \\
 ^\e \xi_2\end{array} \right)={ ^\e \xi}(d,h,y):=\frac{(d^\top,h^\top-\cM^\top(y))^\top}{\e^{1/2}}.
\end{align}
(Generally, for vectors in $\R^{2n}$ or $\R^{n}$, the left superscript as in $^\e \xi$ refers to the scaling factor of $^\e \xi$, whereas right superscripts indicate initial conditions $(d,h,y)$.) We also define the function 
\begin{align}\label{eq:R}
&R^\e(d,h,y,r,p,p',z,q):=-\rho(h^\top d-\frac{d^\top C^{-1}d}{2})-\e\rho r\notag\\
&\qquad\qquad\qquad+\e\Big(-\mu_Y^\top(y)p'-\frac{1}{2}\mathrm{Tr}\left(\sigma_Y(y)\sigma_Y^\top(y)q\right)-\e^{-1/2}\mu_Y^\top  (y) (J_{\bar \cM})^\top  p\\
&\qquad\qquad\qquad\qquad-\e^{-1/2} \sum_{i,j}(\sigma\sigma^\top)_{i,j}(J_\cM z)_{i,j}-\e^{-1/2}\frac{1}{2}\sum_{k=1}^{2d}p_k \mathrm{Tr} (\sigma_Y\sigma_Y^\top   \pa_{yy}{\bar \cM^k})\Big)\notag
\end{align}
where $J$ denote the Jacobian matrix.  

\medskip

 The first step of our proof is to derive an expansion of the action of the operator $G^\e$ from \eqref{dfn:G} on a class of smooth functions.
 
 \begin{prop}\label{thm:remainder} 
Suppose Assumption~\ref{assume-expansion} is satisfied and define
$$\psi^\e(d,h,y):=v(y)-\e\phi(d,h,y)-\e^2\omega\left(^\e \xi(d,h,y),y\right),$$
for smooth functions $v,\phi$ and $\omega$. Then,
  \begin{align}\label{eq:remainder}
   \rho\psi^\e &+ G^\e(d,h,y,\pa_d \psi^\e,\pa_h\psi^\e,\pa_y\psi^\e,\pa_{yy}\psi^\e)\\
 \notag &=\rho v-\cL^Y v -\frac{\mu^\top  \Sigma^{-1}\mu}{2\gamma}-\cH(d,y,-\pa_d\phi,-\pa_h\phi)\\
 \notag  &\quad-\e^{1/2}(\pa_h \phi+C\pa_d\phi)\Lambda^{-1}(\pa_2 w+C\pa_1w)\\
   \notag&\quad+\e R^\e(d,h,y,w(^\e \xi,y),\pa_\xi w(^\e \xi,y),\pa_y w(^\e \xi,y),\pa_{2,y}w(^\e \xi, y),\pa_{yy}w(^\e \xi,y))\\
   \notag& \quad+\e\left(E_1(^\e \xi_1,^\e \xi_2,y,\pa_1 w,\pa_2 w)+\frac{1}{2}\pa_{ \cM\cM}w(^\e \xi, y)-\rho \phi+\cL^Y \phi \right).
  \end{align}
  (Here, all functions of $\xi$ are evaluated at $^\e \xi$ and $\pa_i w$ is the gradient of $w$ in $\xi_i$.) Moreover, there exists a (sufficiently large) constant $c>0$ such that if the function $w$ satisfies
\begin{align}\label{eq:boundew}
|w|+|\pa_y w|+|\pa_{yy} w|+\sqrt{1+|\xi|^2} (|\pa_\xi w|+|\pa_{y\xi}w| )\leq \kappa(y)(1+|\xi|^2)
\end{align}
 for some function $\kappa$, then
 \begin{align}\label{eq:remainderquantitative}
 |R^\e(d,h,y,w(^\e \xi,y),\pa_\xi w(^\e \xi,y),\pa_y w(^\e \xi,y),\pa_{2,y}w(^\e \xi, y),\pa_{yy}w(^\e \xi,y))|\notag\\
 \quad\quad\leq c \kappa(y) \bar M^3(y) \e^{1/2}\sqrt{1+ {^\e \xi}^2}\left(1+\e^{1/2}\sqrt{1+ {^\e \xi}^2}\right).
 \end{align}
  \end{prop}
  
  Note that the inequality \eqref{eq:remainderquantitative} is a quantitative version of the remainder estimate in \cite[Lemma 6.1, (Ri)]{moreau.al.15}. It implies that $R^\e$ from \eqref{eq:R} is bounded on bounded sets of $(d,h,y)$. Moreover, it shows that this remainder term converges to zero on sets of $(d,h,y)$ for which $({^\e \xi},y )$ is bounded. 
  
  \begin{proof}[Proof of Proposition~\ref{thm:remainder}]
  We first compute the required derivatives,
  \begin{align*}
  \pa_{y_i}\psi^\e&= \pa_{y_i} v-\e \pa_{y_i}\phi-\e^2 \pa_{y_i}w+\e^{3/2}\pa_2 w^\top \pa_{y_i} \cM,\\
    \pa_{y_i,y_j}\psi^\e&= \pa_{y_i,y_j}v-\e \pa_{y_i,y_j}\phi-\e^2 \pa_{y_i,y_j}w+\e^{3/2}\pa_{y_i,2}w^\top \pa_{y_j} \cM\\
    &\quad+\e^{3/2}\pa_{2,y_j} w^\top \pa_{y_i} \cM +\e^{3/2}\pa_2 w^\top \pa_{y_i,y_j}M-\e \pa_{y_j}M^\top\pa_{2,2} \pa_{y_i}M,\\
  \pa_{d}\psi^\e&=-\e\pa_d \phi -\e^{3/2}\pa_1 w,\\
  \pa_{h}\psi^\e&=-\e\pa_h \phi -\e^{3/2}\pa_2 w.
  \end{align*}
  As a consequence,
  \begin{align*}
   &\rho\psi^\e+ G^\e(d,h,y,\pa_d \psi^\e,\pa_h\psi^\e,\pa_y\psi^\e,\pa_{yy}\psi^\e)=\rho v-\cL^Y v-\frac{1}{2}\mu^\top (\gamma\Sigma)^{-1}\mu-Rd^\top \pa_d \phi\\
   &\qquad+\e\left(\cL^Y\phi -\rho\phi +f({^\e\xi},y)-R(^\e \xi_1)^\top \pa_1 w-\frac{1}{2}\pa_{\cM\cM}w \right)\\
      &\qquad-\frac{1}{2}\left(\pa_h \phi+C \phi_d+\e^{1/2}(\pa_2 w+C \pa_1 w)\right)^\top\Lambda^{-1}\left(\pa_h \phi+C \phi_d+\e^{1/2}(\pa_2 w+C \pa_1 w)\right)\\
  &\qquad+\e R^\e(d,h,y,w(^\e \xi,y),\pa_\xi w(^\e \xi,y),\pa_y w(^\e \xi,y),\pa_{\xi_2,y}w(^\e \xi,y),\pa_{yy}w(^\e \xi,y))\\
   &=\rho v-\cL^Y v-\frac{1}{2}\mu^\top (\gamma\Sigma)^{-1}\mu-\cH(d,y,-\pa_d\phi,-\pa_h\phi)\\
   &\qquad+\e\left(\cL^Y\phi -\rho\phi +E_1({^\e \xi},y,\pa_1 w,\pa_2 w) +\frac{1}{2}\pa_{ \cM\cM}w\right)\\
   &\qquad-\frac{\e^{1/2}}{2}\left(\pa_h \phi+C \phi_d\right)^\top\Lambda^{-1}\left(\pa_2 w+C \pa_1 w\right)\\
    &\qquad+\e R^\e(d,h,y,w,\pa_\xi w,\pa_y w,\pa_{\xi_2,y}w,\pa_{yy}w).
  \end{align*}

We now prove the inequality \eqref{eq:remainderquantitative} under the additional assumption \eqref{eq:boundew} by dominating each of the terms in the definition \eqref{eq:R} of $R^\e$. Here, the upper bound for $-\rho(h^\top d-\frac{d^\top C^{-1}d}{2})-\e\rho w({^\e \xi},y)$ follows from \eqref{eq:boundew} and the definition of $^\e\xi$. For the remaining terms, we can dominate all terms that only depend on $y$ but not on $w$ or its derivatives by a constant multiple of $\bar M^3(y)$. Whence it remains to estimate the following upper bound,
$$c\bar M^3(y)\left(\e |\pa_{y}w(^\e \xi,y)|+\e |\pa_{yy}w(^\e \xi,y)|+\e^{1/2}|\pa_{\xi_2}w(^\e \xi,y)|+\e^{1/2}|\pa_{\xi_2,y}w(^\e \xi,y)|\right).$$
Taking into account the condition on $w$, this in turn yields the desired upper bound.
 \end{proof}

\subsection{Viscosity Subsolution Property} 
The proof of viscosity properties in this section and the following requires to construct local minima or maxima. Then, we use the viscosity property of $\tilde V^\e$ at these extrema. The construction of the extrema is classical in homogenisation theory and is similar to the proofs of \cite[Proposition 6.3 and Proposition 6.4]{moreau.al.15}. We therefore only outline this construction. In contrast, we give more details on how to use the viscosity property of $\tilde V^\e$, because the quantities that need to be controlled are more involved here. For example, the linear part of the Hamiltonian $\cH$ is new here and its sign needs to be controlled separately.  To ease comparison to the corresponding arguments in  \cite[Proposition 6.3 and Proposition 6.4]{moreau.al.15}, the proofs are broken up into the same steps are there.

\begin{prop}\label{prop:sub}
Suppose the assumptions of Theorem \ref{thm:expansion} are satisfied. Then \eqref{eq:barup} is a viscosity subsolution of the corrector equation \eqref{corrector2}.
\end{prop}

\begin{proof}
We adapt the proof of \cite[Proposition 6.3]{moreau.al.15} to the present setting and keep the same steps for simplicity of reading.  Let $y_0\in \cD$. 

{\it Step 1: Localization.} Under Assumption~\ref{assume-expansion}, Proposition \ref{local-bound-ratios} gives the existence of  $\e_0>0$ and $r_0>0$ such that 
$$b^*:=\sup\left\{ u^\e(d,h,y):(d,h,y)\in B_{r_0}(0,\cM(y_0),y_0),~ \e\in(0,\e_0) \right \}<\infty.$$

Consider a smooth function $\phi$ such that, for all $y\in\cD- \{y_0\}$,
\begin{align}\label{eq:subsol1}
  u^*(y)-\phi(y)&= u^*(0,\cM(y),y)-\phi(y)\notag\\
  &\qquad< u^* (y_0)-\phi(y_0)= u^*(0,\cM(y_0),y_0)-\phi(y_0)=0.
\end{align}
 By the continuity of $\phi$, for all $\e>0$ there exists 
$(d^\e,h^\e,y^\e)\in \R^n\times\R^n\times \cD$, such that 
\begin{gather*}
(d^\e,h^\e,y^\e)\to (0,\cM (y_0),y_0),\\
 u^{*\e}(d^\e,h^\e,y^\e)\to  u^*(0,\cM(y_0),y_0)=  u^* (y_0),\\
p^\e:= u^{*\e}(d^\e,h^\e,y^\e)-\phi (y^\e)\to 0\mbox{ as }\e\to 0. 
\end{gather*}
Recall the constant $\d_0>0$ from Lemma \ref{lem:ricatti}. Similarly to \cite[Proposition 6.3]{moreau.al.15}, taking $\e_0$ smaller if needed, due to the continuity of $\cM$ and the definition of $\varpi$, there exists $\a\in(0,r_0)$ such that with
\begin{align*}
M:=\sup\{2+b^*-\phi(y);y\in B_\a (y_0)\}<\infty, \qquad c_0:=\frac{M}{ ( \d_0^2\wedge 1)(\a/4)^4},
\end{align*}
we have 
\begin{gather*}
|\cM(y)-\cM(y_0)|\leq r_0/4,\mbox{ if }|y-y_0|\leq \a,\\
(|d^\e|+|h^\e-\cM(y_0)|)\vee |y^\e-y_0|\leq \a/4,\quad |p^\e|\leq 1, \\
\notag \varpi^2(d^\e,h^\e-\cM(y^\e),y^\e)\leq \frac{1}{3c_0},\quad  \varpi(d^\e,h^\e-\cM(y^\e),y^\e)\leq \frac{1}{3},
\end{gather*}
for all $\e \in(0,\e_0)$. (Note that the term $\varpi$, which is not present in \cite[Proof of Proposition 6.3]{moreau.al.15}, corrects a minor error in that proof.)  

{\it Step 2: Construction of a test function.} Now define, for $\eta\in (0,1)$,
\begin{align*}
&\Phi^\e:(d,h,y)\to c_0\left(\varpi^2(d,h-\cM(y),y)+|y-y^\e|^4\right),\\
&I^{\e,\eta}(d,h,y):=-u^{*\e}(d,h,y)+p^\e+\phi(y)+\phi^\e(d,h,y)+\eta \varpi(d,h-\cM(y),y),
\end{align*}
and set
\begin{align*}
B_\a:=B_{r_0}(0,\cM(y_0))\times B_{\a}(y_0), \qquad B_{0,\a}:=B_{r_0/2}(0,\cM(y_0))\times B_{\a/2}(y_0).
\end{align*}
On $B_\a- B_{0,\a}$ and for all $\e\in (0,\e_0)$ and $\eta\in (0,1)$, 
$$\Phi^\e\geq M\mbox{ and } I^{\e,\eta}\geq 1.$$
Indeed, if $|(d,h-\cM(y_0))|\geq \frac{r_0}{2}$ then the estimate in Lemma~\ref{lem:ricatti} gives
\begin{align*}
\Phi^\e(d,h,y)&\geq c_0 \varpi^2(d,h-\cM(y),y)\geq c_0 \d^2_0 |(d,h-\cM(y))|^4 \\
&\geq  c_0 \d^2_0 |(d,h-\cM(y_0))|^4-8 c_0 \d^2_0 r_0^4/4^4\geq c_0  \d^2_0 r_0^4(1/2^4 -1/2^5)\\
&\geq c_0  \d^2_0 r_0^4/2^5\geq M.
\end{align*}
 If $|y-y_0|\geq \frac{\a}{2}$, then $|y-y^\e|\geq \frac{\a}{4}$ and in turn $\Phi^\e(d,h,y)\geq c_0 |y-y^\e|^4\geq  c_0 (\a/4)^4\geq M$. Moreover, 
$$\Phi (d^\e,h^\e,y^\e)=c_0(\varpi(d^\e,h^\e-\cM(y^\e)))^2\leq 1/3 \quad \mbox{ and } \quad I^{\e,\eta}(d,^\e,h^\e,y^\e)\leq 2/3.$$
Thus there exists a local minimum $\tilde \theta^\e=(\tilde d^\e,\tilde h^\e,\tilde y^\e)$ of $I^{\e,\eta}$ on the compact set $\overline B_{0,\a}$.
This is equivalent to the fact that $\tilde \theta^\e$ is a local minimum of the function
$$\tilde V^\e_*-\psi^\e$$
where for all $\e\in (0,\e_0)$ and $\eta\in (0,1)$,
$$\psi^\e(d,h,y):=v^0(y)-\e(p^\e+\phi(y)+\Phi^\e(d,h,y)+(1+\eta)\varpi(d,h-\cM(y),y)).$$

Denoting this minimum by $^\e\tilde \xi= {^\e\tilde \xi}:=\left(\frac{\tilde d^e}{\e^{1/2}},\frac{\tilde h^\e-\cM(\tilde y^\e)}{\e^{1/2}}\right)$.

{\it Step 3: Boundedness of $^\e\tilde \xi$} and {\it Step 4: Control of signs.} Using the viscosity supersolution property of $\tilde V^\e_*$ (this is precisely the discontinuous viscosity supersolution property of $\tilde V^\e$), it follows that
  \begin{align}\label{eq:mainsub}
   \notag&\e^{-1}\left(\rho\psi^\e+ G^\e(\tilde d^\e,\tilde h^\e,\tilde y^\e,\pa_d \psi^\e,\pa_h\psi^\e,\pa_y\psi^\e,\pa_{yy}\psi^\e)\right)\\
\notag & =R^\e(\tilde d^\e,\tilde h^\e,\tilde y^\e,(1+\eta)\varpi,(1+\eta)\pa_\xi \varpi,(1+\eta)\pa_{y} \varpi,(1+\eta)\pa_{2,y} \varpi,(1+\eta)\pa_{yy}\varpi)\\
  \notag&\quad-\e^{-1}\cH(\tilde d^\e,\tilde h^\e,-\pa_d\Phi^\e,-\pa_h\Phi^\e)-\e^{-1/2}(1+\eta)(\pa_h \Phi^\e+C\pa_d\Phi^\e)\Lambda^{-1}(\pa_2 \varpi+C\pa_1\varpi)\\
  & \quad+E_1(\tilde \xi_1^\e,\tilde\xi_2^\e,y,(1+\eta)\pa_1 \varpi,(1+\eta)\pa_2 \varpi)+\frac{1+\eta}{2}\pa_{\cM \cM}\varpi-\rho \phi(\tilde y^\e)+\cL^Y \phi(\tilde y^\e) \geq 0
  \end{align}
with $|\tilde d^\e|+|\tilde h^\e-\cM(\tilde y^\e)|\vee|y_0-\tilde y^\e|\leq c'_0$ and where the function $\phi^\e$ is evaluated at $\tilde \theta^\e$ and $ \varpi $ is evaluated at ${^\e\tilde \xi}$.
Note that, due to the boundedness of $\{ \tilde \theta^\e\}_{\{\eta\in (0,1),\e\in (0,\e_0)\}}$, the fact that the second derivative of $y\to  A(y)$ is locally bounded and Proposition \ref{thm:remainder}, the term $$R^\e(\tilde d^\e,\tilde h^\e,\tilde y^\e,(1+\eta)\varpi,(1+\eta)\pa_\xi \varpi,(1+\eta)\pa_{y} \varpi,(1+\eta)\pa_{2,y} \varpi,(1+\eta)\pa_{yy}\varpi)$$ is bounded. By the first corrector equation \eqref{eq:fst}, 
\begin{align*}
&E_1({^\e \tilde \xi},y,(1+\eta)\pa_1 \varpi,(1+\eta)\pa_2 \varpi)+\frac{1+\eta}{2}\pa_{ \cM\cM}\varpi=(1+\eta) a\\
&-\eta\left(\frac{\gamma (^\e \tilde\xi_2)^\top\Sigma {^\e \tilde\xi_2}}{2}+R (^\e \tilde\xi_1)^\top C^{-1} {^\e \tilde\xi_1}\right)-\frac{\eta+\eta^2}{2}(C\pa_1 \varpi +\pa_2 \varpi)^\top  \Lambda^{-1}(C\pa_1 \varpi +\pa_2 \varpi).
\end{align*}
We also compute 
\begin{align*}
&\pa_h \Phi^\e+C\pa_d\Phi^\e=2\e^{3/2}c_0\varpi(^\e\tilde \xi)(C\pa_1 \varpi (^\e\tilde \xi)+\pa_2 \varpi(^\e\tilde \xi)),\\
&\pa_d\Phi^\e=2\e^{3/2}c_0\varpi(^\e\tilde \xi)\pa_1 \varpi ({^\e\tilde \xi}).
\end{align*}
As a consequence,
\begin{align}\label{eq:superremain}
\notag&-\e^{-1}\cH(\tilde d^\e,\tilde h^\e,-\pa_d\Phi^\e,-\pa_h\Phi^\e)-\e^{-1/2}(1+\eta)(\pa_h \Phi^\e+C\pa_d\Phi^\e)\Lambda^{-1}(\pa_2 \varpi+C\pa_1\varpi)\\
\notag&=-\e^{-1}\frac{4\e^3 c_0^2 \varpi^2}{2}(C\pa_1 \varpi +\pa_2 \varpi)^\top  \Lambda^{-1}(C\pa_1 \varpi +\pa_2 \varpi)\\
\notag&\quad+ \e^{-1}\left(-2\e^2 c_0\varpi R(\tilde \xi_1^\e)^\top\pa_1 \varpi \right)  -2 \e(1+\eta)c_0 \varpi(C\pa_1 \varpi +\pa_2 \varpi)^\top  \Lambda^{-1} (C\pa_1 \varpi +\pa_2 \varpi)\\
\notag&=-\left(2\e^2 c_0^2 \varpi ^2+2\e (3/4+\eta)c_0 \varpi\right)(C\pa_1 \varpi +\pa_2 \varpi)^\top  \Lambda^{-1} (C\pa_1 \varpi +\pa_2 \varpi)\\
 &\quad+2\e c_0\varpi\left(-R({^\e\tilde \xi_1})^\top\pa_1 \varpi-\frac{1}{2}(C\pa_1 \varpi +\pa_2 \varpi)^\top  \Lambda^{-1} (C\pa_1 \varpi +\pa_2 \varpi)\right).
\end{align}
Note that by definition of $\varpi$,
$$-R(\tilde \xi_1^\e)^\top\pa_1 \varpi ({^\e\tilde \xi})
-\frac{1}{2}(C\pa_1 \varpi +\pa_2 \varpi)^\top  \Lambda^{-1} (C\pa_1 \varpi +\pa_2 \varpi)=-f(\tilde \xi^\e,\tilde y^\e)\leq 0,$$
which implies that 
$$-\e^{-1}\cH(\tilde d^\e,\tilde h^\e,-\pa_d\Phi^\e,-\pa_h\Phi^\e)-\e^{-1/2}(1+\eta)(\pa_h \Phi^\e+C\pa_d\Phi^\e)\Lambda^{-1}(\pa_2 \varpi+C\pa_1\varpi)\\
\leq 0.$$
Thus, together with \eqref{eq:degenere-source}, it follows that
\begin{align*}
&R^\e(\tilde d^\e,\tilde h^\e,\tilde y^\e,(1+\eta)\varpi,(1+\eta)\pa_\xi \varpi,(1+\eta)\pa_{y} \varpi,(1+\eta)\pa_{2,y} \varpi,(1+\eta)\pa_{yy}\varpi)\\
& \quad+(1+\eta) a(\tilde y^\e)-\rho \phi(\tilde y^\e)+\cL^Y \phi (\tilde y^\e)\geq \eta\left(\frac{\gamma ({^\e \tilde\xi_2})^\top\Sigma {^\e \tilde\xi_2}}{2}+R ({^\e \tilde\xi_1})^\top C^{-1} {^\e \tilde\xi_1}\right)\\
&\geq\eta \inf\{m(\tilde y^\e),\e\in (0,\e_0)\}|{^\e\tilde \xi}|^2
\end{align*}
so that the family $\{({^\e \tilde\xi_1},{^\e \tilde\xi_2},y^\e),\e\in(0,\e_0)\}$ is bounded.

{\it Step 5: Conclude.} Let $(\tilde \xi_1,\tilde \xi_2,\bar y)$ be an accumulation point of this family (which might depend on $\eta\in (0,1)$). 
The strict inequality \eqref{eq:subsol1} implies $\bar y=y_0$. In addition, the boundedness of $ \{({^\e \tilde\xi},y^\e),\e\in(0,\e_0)\}$ combined with \eqref{eq:superremain} and Proposition~\ref{thm:remainder} gives
\begin{align*}
-\e^{-1}\cH(\tilde d^\e,\tilde h^\e,-\pa_d\phi^\e,-\pa_h\phi^\e)-\e^{-1/2}(1+\eta)(\pa_h \phi^\e+C\pa_d\phi^\e)\Lambda^{-1}(\pa_2 \varpi+C\pa_1\varpi)\to 0,\\
R^\e(\tilde d^\e,\tilde h^\e,\tilde y^\e,(1+\eta)\varpi,(1+\eta)\pa_\xi \varpi,(1+\eta)\pa_{y} \varpi,(1+\eta)\pa_{2,y} \varpi,(1+\eta)\pa_{yy}\varpi)\to 0,
\end{align*}
as $\e\to 0$. We finally use \eqref{eq:mainsub} to obtain that, for all $\eta\in (0,1)$,
$$
E_1(\tilde \xi_1,\tilde\xi_2,y_0,(1+\eta)\pa_1 \varpi,(1+\eta)\pa_2 \varpi)+\frac{1+\eta}{2}\pa_{ \cM\cM}\varpi(y_0)-\rho \phi(\tilde y_0)+\cL^Y \phi(\tilde y_0) \geq 0.
$$
(Here, $\tilde \xi_1$ might depend on $\eta$.) Using the first corrector equation \eqref{eq:fst} and sending $\eta$ to $0$, we obtain
\begin{align*}
& a(y_0)-\rho \phi( y_0)+\cL^Y \phi( y_0) \geq\\
&\qquad\qquad\liminf_{\eta\to 0}\left(\eta\left(\frac{\gamma (\tilde\xi_2)^\top\Sigma \tilde\xi_2}{2}+R (\tilde\xi_1)^\top C^{-1} \tilde\xi_1\right)\right.\\
&\qquad\qquad\left.+\frac{\eta(1+\eta)}{2}(C\pa_1 \varpi +\pa_2 \varpi)^\top  \Lambda^{-1} (C\pa_1 \varpi +\pa_2 \varpi)\right)\geq 0.
\end{align*}
This establishes the claimed viscosity subsolution property.
\end{proof}

\subsection{Viscosity Supersolution Property}
\begin{prop}\label{prop:super}
Suppose the assumptions of Theorem \ref{thm:expansion} are satisfied. Then \eqref{eq:bardown} is a viscosity supersolution of the corrector equation \eqref{corrector2}.
\end{prop}

\begin{proof}
As mentioned in Remark \ref{rem:ham} our Hamiltonian $\cH$ is degenerate and does not satisfy the non-degeneracy condition \cite[Equation 6.27]{moreau.al.15} that plays a crucial role in the asymptotic analysis of the model with only temporary trading costs. We therefore outline how to modify the proof of the supersolution property in \cite[Proposition 6.4]{moreau.al.15} to be able to use the non-degeneracy of our source term $f$ rather than of $\cH$.

We start similarly to the subsolution property. By Proposition \ref{local-bound-ratios}, $u_*$ is lower semicontinuous, non-negative and locally bounded.  Let $y_0\in \cD$ and $\phi$ smooth such that for all $y\in\cD- \{y_0\}$ we have
\begin{align}\label{eq:supersol1}
 u_*(y)-\phi(y)= u_*(0,\cM(y),y)-\phi(y)&> u_* (y_0)-\phi(y_0)\notag\\
 &= u_*(0,\cM(y_0),y_0)-\phi(y_0)=0.
\end{align}
 By the continuity of $\phi$, for all $\e>0$ there exists $(d^\e,h^\e,y^\e)\in \R^n\times\R^n\times \cD$ such that 
\begin{gather*}
(d^\e,h^\e,y^\e)\to (0,\cM (y_0),y_0),\\
 u^\e_*(d^\e,h^\e,y^\e)\to  u_*(0,\cM(y_0),y_0)=  u_* (y_0),\\
p^\e:= u^\e_*(d^\e,h^\e,y^\e)-\phi (y^\e)\to 0\mbox{ as }\e\to 0. 
\end{gather*}
Similarly to \cite{moreau.al.15}, we can take $r_0,\e_0>0$ small enough such that, for all $\e\in (0,\e_0)$,
\begin{align}
|y^\e-y_0|\leq r_0/2,\qquad  \varpi(d^\e,h^\e-\cM(y^\e),y^\e)\leq \frac{1}{3},\qquad |p^\e|\leq 1.
\end{align}

{\it Step 1: Penalize $\phi^\e$.} Set
\begin{align*}
M:=\sup\{4+\phi(y);y\in B_{r_0} (y_0)\}, \quad c_0:=\frac{2^4 M}{r_0^4}, \quad \phi^\e(y):=\phi(y)+p^\e-c_0 |y-y^\e|^4.
\end{align*}
By the choice of $c_0$, we have
$$\phi^\e(y)\leq -3, \quad  \mbox{for all $y\notin B_{r_0}(y_0)$ and $\e\in(0,\e_0)$,}$$
and 
\begin{align}\label{eq:controlphie}
0=- u^\e_*(d^\e,h^\e,y^\e)+\phi^\e(y^\e)\geq-\frac{1}{3}.
\end{align}
Similarly as in \cite{moreau.al.15}, define 
\begin{align*}
K_0&:=\sup\{{|\rho\phi-\cL^Y \phi (y)|}:y\in B_{r_0}(y_0)\}<\infty,\\
K_2&:=\sup\{{| A(y)|}:y\in B_{r_0}(y_0)\}<\infty,\\
K_{\bar M}&:=\sup\{1+|\bar M|^4:y\in B_{r_0}(y_0)\}<\infty.
\end{align*}
 For all $\eta\in (0,1]$, pick a function $h\in C^\infty (\R^{2n};[0,1])$ and $a_\eta\in (1,\infty)$ such that 
\begin{gather*}
h^\eta(\xi)=1 \mbox{ if }|\xi|\leq 1, \quad h^\eta(\xi)=0\mbox{ if }|\xi|\geq a_\eta,\\
|\xi||\pa_\xi h^\eta(\xi)|\leq \eta, \quad \mbox{ and } \quad |\xi|^2|\pa_{\xi\xi}h^\eta(\xi)|\leq C^*\mbox{ for all }\xi\in \R^{2n},
\end{gather*}
for some $C^*>0$. We write, for all $\d>0$,
$$\xi^{*,\d}:=\sqrt{1+ \frac{2(K_0+(1-\d)K_{\bar M}K_2(C^*+6))}{\d  \underline m}}$$
and 
$$H^{\eta,\delta}(\xi):=(1-\delta)h^\eta\left(\frac{\xi}{\xi^{*,\d}}\right).$$

{\it Step 2 and 3: Construct two test functions.}
Now define 
$$\Psi^{\eta,\delta}(d,h,y):=V^0(y)-\e \phi^\e(y)-\e^2 H^{\eta,\delta}({^\e \xi})\varpi({^\e \xi })$$
and -- similarly to \cite[Proof of Proposition 6.4]{moreau.al.15} -- fix an even, smooth function $F$ satisfying 
$$F(0)=1,\qquad F(x)=0\mbox{ if }|x|\geq 1,\qquad 1\geq F\geq 0.$$
Then, for all $(\e,\eta,\d)\in (0,\e_0]\times(0,1]\times(0,1]$ there exist $(\hat d^{\e,\eta,\d},\hat h^{\e,\eta,\d},\hat y^{\e,\eta,\d})$ with $|\hat y^{\e,\eta,\d}-y_0|< r_0$ satisying the following two properties: i) there are $(\tilde d^{\e,\eta,\d},\tilde h^{\e,\eta,\d},\tilde y^{\e,\eta,\d})$ such that 
\begin{align*}
|\tilde d^{\e,\eta,\d}-\hat d^{\e,\eta,\d}|^2+|\tilde h^{\e,\eta,\d}-\hat h^{\e,\eta,\d}|^2\leq 1, \qquad |\tilde y^{\e,\eta,\d}-\hat y^{\e,\eta,\d}|\leq r_0,
\end{align*}
and ii) the function 
\begin{align*}
&\tilde V^{*,\e}(d,h,y)-\bar \psi^{\e,\eta,\d}(d,h,y)\\
&:=\tilde V^{*,\e}(d,h,y)- \psi^{\e,\eta,\d}(d,h,y)+\e^2 F\left(\sqrt{|d-\hat d^{\e,\eta,\d}|^2+|h-\hat h^{\e,\eta,\d}|^2}\right)
\end{align*}
has a local minimum at $(\tilde d^{\e,\eta,\d},\tilde h^{\e,\eta,\d},\tilde y^{\e,\eta,\d})$. 

{\it Step 4: Boundedness of $\{{^\e \tilde \xi^{\eta,\d}}\}$ for suitable $\e>0$.}
Defining $\e_{\eta,\d}:=\e_0\wedge (K_2^{1/2}a_\eta \xi^{*,\d})^{-1}$ and, similarly to \cite{moreau.al.15}, using the discontinuous viscosity subsolution property of $\tilde V^\e$, it follows that
the family $\{{^\e \tilde \xi^{\eta,\d}}:\e \in(0,\e_{\eta,\d})\}$ is bounded and, up to taking a subsequence, we have the convergence $({^\e \tilde \xi^{\eta,\d}},\tilde y^{\e,\eta,\d})\to ({\bar \xi^{\eta,\d}},\bar y^{\eta,\d})$ as $\e\to 0$. 

{\it Step 5: Boundedness of $\{\bar \xi^{\eta,\d}\}$ for suitable $\eta>0$.}
Using this boundedness and Proposition~\ref{thm:remainder} similarly as in \cite[Proof of Proposition 6.4]{moreau.al.15}, we obtain 
$$
E_1({\bar \xi^{\eta,\d}},\bar y^{\eta,\d},\pa_1 (H^{\eta,\d}\varpi),\pa_2 (H^{\eta,\d}\varpi))+\frac{1}{2}\pa_{ \cM\cM} (H^{\eta,\d}\varpi)-\rho \phi(\bar y^{\eta,\d})+\cL^Y \phi(\bar y^{\eta,\d})\leq 0.
$$
The definitions of $K_0,K_2$ and $K_{\bar M}$ imply
\begin{align*}
\left|\frac{1}{2}\pa_{ \cM\cM} (H^{\eta,\d}\varpi)-\rho \phi(\bar y^{\eta,\d})+\cL^Y \phi(\bar y^{\eta,\d})\right|\leq K_0+(1-\d)K_{\bar M}K_2(C^*+6).
\end{align*}
As a consequence,
$$E_1({\bar \xi^{\eta,\d}},\bar y^{\eta,\d},\pa_1 (H^{\eta,\d}\varpi),\pa_2 (H^{\eta,\d}\varpi))\leq K_0+(1-\d)K_{\bar M}K_2(C^*+6).$$
We can expand the left-hand side of this estimate as follows,
\begin{align*}
&\left(1-(1-\d) h^\eta   \right) f({\bar \xi^{\eta,\d}},\bar y^{\eta,\d})+(1-\d) h^\eta    E_1({\bar \xi^{\eta,\d}},\bar y^{\eta,\d},-\pa_1  \varpi ,-\pa_2  \varpi )\\
&\quad+\frac{(1-\d) h^\eta   -(1-\d)^2(h^\eta)^2}{2}(C\pa_1 \varpi +\pa_2 \varpi)^\top  \Lambda^{-1} (C\pa_1 \varpi +\pa_2 \varpi)\\
&\quad-(1-\d)\frac{R \varpi  }{\xi^{*,\d}} \left(\pa_1  h^\eta   \right)^\top \bar \xi_1^{\eta,\d}\\
&\quad-(1-\d)^2\left(\frac{ h^\eta    \varpi }{\xi^{*,\d}}\left(C\pa_1  h^\eta   +\pa_2  h^\eta   \right)^\top  \Lambda^{-1} (C\pa_1 \varpi +\pa_2 \varpi) \right)\\
&\quad-(1-\d)^2\left(\frac{ \varpi ^2}{(\xi^{*,\d})^2}(C\pa_1  h^\eta   +\pa_2  h^\eta   )^\top  \Lambda^{-1} \left(C\pa_1  h^\eta   +\pa_2  h^\eta   \right) \right),
\end{align*}
where $h^\eta$ is evaluated at $\frac{\xi}{\xi^{*,\d}}$ and $\varpi$ is evaluated at $ ({\bar \xi^{\eta,\d}},\bar y^{\eta,\d})$.
Due to the non-degeneracy of the source term $f$, this can be bounded from below by
\begin{align*}
&\d \underline m|{\bar \xi^{\eta,\d}}|^2 -(1-\d)^2\left(\frac{ h^\eta    \varpi }{\xi^{*,\d}}\left(C\pa_1  h^\eta   +\pa_2  h^\eta   \right)^\top  \Lambda^{-1} (C\pa_1 \varpi +\pa_2 \varpi) \right)\\
&-(1-\d)\frac{R \varpi  }{\xi^{*,\d}} \left(\pa_1  h^\eta   \right)^\top \bar \xi_1^{\eta,\d}-(1-\d)^2\frac{ \varpi ^2}{(\xi^{*,\d})^2}(C\pa_1  h^\eta   +\pa_2  h^\eta   )^\top  \Lambda^{-1} \left(C\pa_1  h^\eta   +\pa_2  h^\eta   \right) .
\end{align*}
The conditions on $h^\eta$ and the definition of $K_2$ in turn yield
\begin{align*}
\left|\frac{R \varpi  }{\xi^{*,\d}} \left(\pa_1  h^\eta   \right)^\top \bar \xi_1^{\eta,\d}\right| &\leq R\eta K_2  |\bar \xi^{\eta,\d}|^2, \\
\left|(1-\d)\frac{ h^\eta    \varpi }{\xi^{*,\d}}\left(C\pa_1  h^\eta   +\pa_2  h^\eta   \right)^\top  \Lambda^{-1} (C\pa_1 \varpi +\pa_2 \varpi) \right| &\leq (1-\d) | \hat C \Lambda^{-1}\hat C^\top |K_2^2 \eta  |\bar \xi^{\eta,\d}|^2,\\
\left|(1-\d)\frac{ \varpi ^2}{(\xi^{*,\d})^2}(C\pa_1  h^\eta   +\pa_2  h^\eta   )^\top  \Lambda^{-1} \left(C\pa_1  h^\eta   +\pa_2  h^\eta   \right) \right| &\leq (1-\d) | \hat C \Lambda^{-1}\hat C^\top |K_2^2 \eta^2|\bar \xi^{\eta,\d}|^2.
\end{align*}
Therefore, we obtain the following bound,
\begin{align*}
 & K_0+(1-\d)K_{\bar M}K_2(C^*+6)\\
  &\geq \left(\d \underline m-(1-\d)K_2  \eta \left(R+(1+\eta)(1-\d)K_2 | \hat C \Lambda^{-1}\hat C^\top |\right)\right)|{\bar \xi^{\eta,\d}}|^2.
  \end{align*}
Note that the last term on the right-hand side goes to $0$ as $\eta \to 0$. This implies that for all $\d\in(0,1)$ there exists $\eta_\d\in (0,1)$ such that, for all $\eta\in(0,\eta_\d)$,
$$
(\xi^{*,\d})^2\geq  \frac{2(K_0+(1-\d)K_{\bar M}K_2(C^*+6))}{\d  \underline m}\geq |{\bar \xi^{\eta,\d}}|^2.
$$
We can now proceed as in \cite[Step 6, Proof of Proposition 6.4]{moreau.al.15} to show that 
$$0\geq a(y_0)-\rho \phi(\tilde y_0)+\cL^Y \phi(\tilde y_0),$$
which is the desired supersolution property. 
\end{proof}

By combining Propositions~\ref{prop:sub} and \ref{prop:super} with the comparison principle from Assumption~\ref{assume-comparison} and the upper bound for the semilimits at Proposition {\ref{local-bound-ratios}}, we obtain the main result of this section.

\begin{thm}\label{thm:cm}
Suppose the assumptions of Theorem \ref{thm:expansion} are satisfied. Then,
\begin{align}
  u^*(0,\cM(y),y)=  u_*(0,\cM(y),y), \quad \mbox{for all $y\in \cD$}.
\end{align}
\end{thm}

\section{Dependence of $ u^*$ and $ u_*$ on $(d,h)$}\label{s.first}

Recall the source term $ f$ and the Hamiltonian $\cH$ from \eqref{eq:source}, and define
\begin{align}
 u(d,h,y):=  u^*(0,\cM(y),y) + \varpi({^1\xi},y)= u_*(0,\cM(y),y) + \varpi({^1\xi},y).
\end{align}
As $\varpi$ is a second order polynomial in ${^1\xi}$, the function $u$ is a smooth solution of the following first-order PDE,
\begin{align}\label{eq:eikhat}
 f(d,h-\cM(y),y)=  \cH(d,h-\cM(y),-\pa_d  u(d,h,y),-\pa_d  u(d,h,y)).
\end{align}

The goal of this section is to prove the following result.

\begin{prop}\label{prop:conclusioneik}
The upper and lower semilimits $ u^*$ and $ u_*$ from \eqref{eq:semi-limits-bar} satisfy
\begin{align}
 u^*(\theta)\leq u(\theta)\leq  u_*(\theta), \quad \mbox{for all $\theta$}.
\end{align}
\end{prop}

The converse inequality evidently holds by definition of $u^*$ and $ u_*$. Whence, Proposition~\ref{prop:conclusioneik} shows that all three functions are equal and depend on the initial conditions $(d,h)$ of the price distortion and the risky positions through the function $\varpi$. We will establish this result by first showing that the semilimits $ u^*$ and $ u_*$ are viscosity sub- and supersolutions, respectively, of the first-order PDE \eqref{eq:eikhat}. We then conclude by proving that, under the condition $u^*(0,\cM(y),y)\leq u_*(0,\cM(y),y)$ (which we have already verified, cf.~Theorem~\ref{thm:cm}), this PDE admits a comparison result among non-negative semisolutions.

\subsection{A First-Order Equation}

We first prove that the semilimits $ u^*$ and $ u_*$ are viscosity sub- and supersolution of \eqref{eq:eikhat}:

\begin{lem}\label{lem:eik}
For all $y\in \cD$, the function $(d,h)\to  u^* (d,h,y)\geq 0$ is a viscosity subsolution of
\begin{align}\label{eq:eik}
 f(d,h-\cM(y),y)\geq \cH(d,h-\cM(y),-\pa_d  u^*(d,h,y),-\pa_d  u^*(d,h,y)).
\end{align}
Likewise, $(d,h)\to  u_* (d,h,y)$ is a viscosity supersolution of
\begin{align}
 f(d,h-\cM(y),y)\leq  \cH(d,h-\cM(y),-\pa_d  u_*(d,h,y),-\pa_d  u_*(d,h,y)).
\end{align}
\end{lem}

\begin{proof}
We only prove the subsolution property; the supersolution property can be verified along the same lines. Consider a smooth function $\phi$ and $\theta_0=(d_0,h_0,y_0)\in\R^n\times\R^n\times \cD$ such that the following strict local maximality holds,
$$(  u^*-\phi)(\theta)<(  u^*-\phi)(\theta_0) =0, \quad \mbox{for all $\theta\neq \theta_0$.}$$
There exists a family $\{\theta_\e\}_{\e > 0}$ satisfying the following properties,
\begin{gather*}
\theta_\e\to \theta_0,\quad  u^{*,\e}(\theta_\e)\to  u^*(\theta_0),\\
\mbox{ and }p^\e:=u^{*,\e}(\theta_\e)- \phi(\theta_\e)\to 0.
\end{gather*}
Similarly to \cite[Lemma 6.7]{moreau.al.15} there are $c_0,r_0>0$ and  $\tilde \theta_\e\to\theta_0$ such that the function
\begin{align}
(v^\e-\psi^\e)(\theta)=\tilde V^\e_*(\theta)-v^0(y)-\e(p^\e+ \phi(\theta)+c_0(|d-d_\e|^4+|h-h_\e|^4+|y-y_\e|^4))
\end{align} 
has a local minimum at $\tilde \theta_\e\in B_{4r_0}(0,\cM(y_0),y_0)$.

Using the discontinuous viscosity supersolution property of $\tilde V^\e$ and Proposition \ref{thm:remainder}, it follows similarly as in \cite[Lemma 6.7]{moreau.al.15} that
\begin{align*}
 0& \leq - \cH(\tilde d^\e,\tilde h^\e,\tilde y^\e,-\pa_d\phi^\e,-\pa_h\phi^\e)-\e\rho((\tilde h^\e)^\top \tilde d^\e-\frac{\tilde (d^\e)^\top C^{-1}\tilde d^\e}{2})\\
  &\quad +\e\left( f (^\e\tilde \xi_1,^\e\tilde  \xi_2,\tilde y^\e)+\frac{1}{2}\pa_{\bar \cM\bar\cM}w(^\e \tilde \xi, \tilde y^\e)-\rho \phi^\e(\tilde y^\e)+\cL^Y \phi^\e(\tilde y^\e) \right)\\
  &= f (\tilde d^\e,\tilde  h^\e-\cM(\tilde y^\e),\tilde y^\e) -\cH(\tilde d^\e,\tilde h^\e,\tilde y^\e,-\pa_d\phi^\e,-\pa_h\phi^\e)\\
  &\quad+\e\left(-\rho((\tilde h^\e)^\top \tilde d^\e+\frac{\tilde (d^\e)^\top C^{-1}\tilde d^\e}{2})+\frac{1}{2}\pa_{ \cM\cM}w(^\e \tilde \xi, \tilde y^\e)-\rho \phi^\e(\tilde y^\e)+\cL^Y \phi^\e(\tilde y^\e) \right).
\end{align*}
Due to the boundedness of $\tilde \theta^\e$ and the definition of $\phi^\e$, the last line goes to $0$ as $\e \to 0$. Thus, for $\e \to 0$, we obtain the asserted viscosity subsolution property.
 \end{proof}
 
For the rest of the section, we fix $y\in \cD$ and omit the dependence on $y$ to ease notation. Note that the first-order equation \eqref{eq:eikhat} lacks ``properness'' in the sense of \cite[Page 2]{CrIsLi92}. To restore properness and prove a comparison result, we define the following auxiliary functions,
\begin{align*}
\bar u^*(d,h):=-e^{- u^*(d,h)}\in [-1,0), \qquad \bar u_*(d,h):=-e^{- u_*(d,h)}\in[-1,0),
\end{align*}
as well as 
\begin{align}\label{dfn:tilde-varpi}
\bar u (d,h)=-e^{- u(d,h)}\in[-1,0).
\end{align}
We also introduce the generator
\begin{align}
\bar H(d,h,r,p_1,p_2)=-r^2  f(d,h-\cM)+ \frac{1}{2}(p_2+C p_1)\top\Lambda^{-1}(p_2+C p_1)-rRd^\top p_1.
\end{align}
Using the viscosity properties from Lemma \ref{lem:eik}, direct computations show that $\bar u^*$ and $\bar u_*$ are viscosity sub- and supersolutions of the PDE corresponding to $\bar{H}$,
\begin{align}\label{eq:eik-proper}
0\geq \bar H(\cdot,\bar u^*, \pa_d \bar u^* , \pa_h \bar u^* ),\\
0\leq \bar H(\cdot,\bar u_*, \pa_d \bar u_*, \pa_h \bar u_* ),\label{eq:eik-proper-super}
\end{align}
and $\bar u$ is smooth solutions of the same PDE,
\begin{align}\label{eq:eikbar}
0= \bar H(\cdot,\bar u, \pa_d \bar u , \pa_h  \bar u ).
\end{align}
Moreover,
\begin{align*}
\bar u(0,\cM)=\bar u^*(0,\cM)=\bar u_*(0,\cM).
\end{align*}

The next step is to obtain a result enabling us to compare the semisolutions $\bar u^*$ and $\bar u_*$ of the PDE \eqref{eq:eikbar}. Note that our task here is simpler than establishing a general comparison result because the PDE admits a smooth solution $\bar u$. Therefore it is sufficient to obtain the following ``partial comparison result'' inspired by \cite[Proposition 5.3]{ETZ1}, \cite[Lemma 5.7]{EKTZ} and \cite[Proof of Lemma 6.10]{moreau.al.15}:

\begin{lem}\label{lem:compeik}
There exists a partial comparison result for bounded, nonpositive viscosity semisolutions of \eqref{eq:eikbar} in the following sense.
Let $v_1$ a bounded, lower semicontinuous, nonpositive viscosity supersolution  of \eqref{eq:eik-proper-super} and $v_2$ a bounded, upper semicontinuous, nonpositive viscosity subsolution of \eqref{eq:eik-proper}  satisfying 
\begin{equation}\label{eq:asscomp}
v_1(0,\cM)\geq v_2(0,\cM).
\end{equation}
If either one of $v_1$ or $v_2$ is continuously differentiable with bounded derivatives then
$$v_1(d,h)\geq v_2(d,h), \quad \mbox{ for all }(d,h)\in\R^n\times \R^n.$$
\end{lem}

\begin{proof}
We focus on the case where $v_2$ is continuously differentiable with bounded derivatives. The other case can be treated similarly.  Assume that the comparison does not hold and 
$$-\a:=v_1(\bar d,\bar h)-v_2(\bar d,\bar h)<0\mbox{ for some }(\bar d,\bar h)\in\R^n\times \R^n.$$
Set $M=\sup\{|v_1|\}+\sup\{|v_2|\}<\infty$. We proceed similarly as in \cite[Proof of Lemma 6.10]{moreau.al.15} and fix $\beta\in C^\infty(\R^n\times \R^n)$, satisfying $0\leq \beta\leq 1$, $\beta(0,0)=1$, $\pa_{d} \beta(0,0)=\pa_h \beta(0,0)=0$, and $\beta(d,h)=0$ for $(d,h)\notin B_1(0,0)$. We can also ensure that $\beta$ satisfies the following non-degeneracy and monotonicity conditions,
$$\{\beta=1\}=\{(0,0)\} \quad \mbox{and} \quad \beta\left( d, h\right) \leq  \beta\left(\frac{(d,h)}{\eta}\right) \quad \mbox{for all $\eta\geq 1$.}$$ 
For $n>1$, define
$$\Phi_n (d,h):=(v_1-v_2-2 M\beta_n(\cdot -(\bar d,\bar h)))(d,h), \quad \mbox{where }\beta_n(d,h)=\beta\left(\frac{(d,h)}{n}\right).$$
Note that $\Phi_n (d,h)=(v_1-v_2)(d,h)\geq -2M$ and $\Phi_n (\bar d,\bar h)=-\a-2M<-2M$ for all $(d,h)\notin  B_n(\bar d,\bar h)$. Thus, there exists $ ( d_n, h_n)\in B_n(\bar d, \bar h)$ such that 
\begin{align*}
&v_1(d_n,h_n)-v_2(d_n,h_n)-2 M\beta_n( (d_n-\bar d,h_n-\bar h))\\
&=\Phi_n (d_n,h_n)= \inf_{d,h} \Phi_n (d,h)\leq \Phi_n \left(\bar d,\bar h\right)\leq -\a-2M .
\end{align*}
Note that this implies 
\begin{align}\label{sign-diff}v_1(d_n,h_n)-v_2(d_n,h_n)\leq -\a.\end{align}
As $( d_n, h_n)\in B_n(\bar d, \bar h)$, up to taking a subsequence, we have 
$$\frac{(d_{n}-\bar d,h_{n}-\bar h)}{{n}}\to (d^*,h^*)\in B_1(0), \quad\mbox{as $n\to \infty$}.$$
By monotonicity of $\beta_n$ in $n$, for all $n\geq n'\geq 1$ we have $\Phi_{n} (d,h)\leq \Phi_{n'} (d,h)\mbox{ for all }(d,h)$. Hence, 
$$\Phi_{{n+1}} (d_{{n+1}},h_{{n+1}}) \leq \Phi_{{n}} (d_{{n}},h_{{n}})\leq -\a-2M,$$
and in turn
$$ \lim_{n \to \infty} \Phi_n (d_{{n}},h_{{n}})=\lim_{n \to \infty} (v_1-v_2)(d_{{n}},h_{{n}})-2 M\beta\left(d^*,h^*\right)\downarrow \Phi^* \leq -\a-2M.$$

We now use the function 
$$
(d,h)\to \tilde \phi_{\eta}(d,h):=v_2(d,h)+2 M\beta_n(d-\bar d,h-\bar h)
$$
as a test function for $v_1$ at $( d_{n}, h_{n})$ to obtain 
$$0\leq \bar H\left( d_{n}, h_{n},v_1( d_{n}, h_{n}),\pa v_2( d_{n}, h_{n})+\frac{2M}{n}\pa \tilde \beta\left( \frac{d_{n}-\bar d, h_{n-\bar h}}{n}\right)\right).$$
As $v_2$ is a smooth subsolution of \eqref{eq:eik-proper} we have 
$$0 \leq -\bar H\left( d_{n}, h_{n},v_2( d_{n}, h_{n}),\pa v_2( d_{n}, h_{n})\right)$$ 
which implies
\begin{align*}
&0\leq  \bar H\left( d_{n}, h_{n},v_1( d_{n}, h_{n}),\pa v_2( d_{n}, h_{n})+\frac{2M}{n}\pa \tilde \beta_{n}( d_{\eta}, h_{\eta})\right)\\
&\leq \bar H\left( d_{n}, h_{n},v_1( d_{n}, h_{n}),\pa v_2( d_{n}, h_{n})+\frac{2M}{n}\pa \tilde \beta_{n}( d_{\eta}, h_{\eta})\right)\\
&\qquad-\bar H\left( d_{n}, h_{n},v_2( d_{n}, h_{n}),\pa v_2( d_{n}, h_{n})\right)\\
&\leq  (v_2^2(d_{n}, h_{n})-v_1^2(d_{n}, h_{n}))  f(d_{n}, h_{n}-\cM) -R\bar d_n^\top \pa_d  v_2(\cdot) (v_1(\cdot)-v_2(\cdot))(d_{n}, h_{n})\\
&+\frac{2M^2}{n^2}\pa \beta\left( \frac{d_{n}-\bar d, h_{n}-\bar h}{n}\right)^\top\hat C\Lambda^{-1}\hat C\pa \beta\left( \frac{d_{n}-\bar d, h_{n}-\bar h}{n}\right)\\
&-2MR \frac{d_n^\top}{n}\pa \beta\left( \frac{d_{n}-\bar d, h_{n}-\bar h}{n}\right)+\frac{2M}{n}\pa \beta\left( \frac{d_{n}-\bar d, h_{n}-\bar h}{n}\right)^\top\hat C\Lambda^{-1}\hat C\pa v_2\left( {d_{n} ,h_{n}}\right).
\end{align*}
Note that one can find $c>0$ and $r_n\to 0$ such that the last two lines are bounded by 
$$c\frac{ \left|(d_n,h_n-\cM)\right|}{n}+r_n.$$
Additionally, using again the viscosity subsolution property of $v_2$ and its sign we find that 
$${- v_2(d_{n}, h_{n}) f(d_{n}, h_{n}-\cM)}-R\bar d_n^\top \pa_d  v_2(d_{n}, h_{n})\geq  -\frac{\pa v_2(\cdot)^\top\hat C\Lambda^{-1}\hat C\pa v_2(\cdot)}{2v_2(\cdot) }\left( {d_{n}, h_{n}}\right) \geq 0.$$
Combining this with \eqref{sign-diff} we obtain 
\begin{align*}
c\frac{ \left|(d_n,h_n-\cM)\right|}{n}+r_n&\geq  (v_2(\cdot)-v_1(\cdot)) \left\{- (v_2(\cdot)+v_1(\cdot))  f(d_{n}, h_{n}-\cM)\right.\\
&\left. \qquad-R\bar d_n^\top \pa_d  v_2(\cdot) \right\}\left( {d_{n}, h_{n}}\right) \\
&\geq  (v_2(d_{n}, h_{n})-v_1(d_{n}, h_{n})) \left\{- v_1(d_{n}, h_{n})  f(d_{n}, h_{n}-\cM) \right.\\
&\left.\qquad -\frac{1}{2v_2(d_{n}, h_{n}) }\pa v_2\left( {d_{n}, h_{n}}\right)^\top\hat C\Lambda^{-1}\hat C\pa v_2\left( {d_{n} h_{n}}\right)\right\}\\
&\geq  - v_1(d_{n}, h_{n}) (v_2(d_{n}, h_{n})-v_1(d_{n}, h_{n}))  f(d_{n}, h_{n}-\cM).
\end{align*}
Together with \eqref{sign-diff}, it follows that
\begin{align*}
c\frac{ \left|(d_n,h_n-\cM)\right|}{n}+r_n\geq \a^{2} f(d_{n}, h_{n}-\cM)\geq \a^2\underline m \left|(d_n,h_n-\cM)\right|.
\end{align*}
Hence, $\left|(d_n,h_n-\cM)\right|\to 0$ as $n\to \infty$. Using one more time \eqref{sign-diff} and the lower semicontinuity of $v_1-v_2$ finally yields 
$$-\a\geq \liminf_n v_1(d_n,h_n)-v_2(d_n,h_n)\geq  v_1(0,\cM)-v_2(0,\cM)\geq 0,$$
which is a contradiction to \eqref{eq:asscomp}. Whence, comparison holds for \eqref{eq:eikbar} under this assumption as asserted.
\end{proof}

\begin{proof}[Proof of Proposition \ref{prop:conclusioneik}]
Recall the viscosity properties stated in \eqref{eq:eik-proper}, \eqref{eq:eik-proper-super} and \eqref{eq:eikbar}. Apply Lemma~\ref{lem:compeik} successively to the pairs $\bar u^*,\bar u $ and $\bar u ,\bar u_*$, which are all equal at $(0,\cM(y),y)$. This yields the following inequalities, 
$$\bar u^*(\theta)\leq \bar u(\theta)\leq \bar u_*(\theta), \quad \mbox{for all $\theta\in \R^n\times \R^n \times \cD$}.$$
As a consequence,
$$ u^*(\theta)\leq  u(\theta)\leq  u_*(\theta).$$
By definition of $u^*$ and $u_*$, all three functions are indeed equal as claimed.
\end{proof}

 \section{Asymptotically Optimal Portfolios}\label{s.optimal.policy}
 
 We now turn to the proof of our second main result, Theorem~\ref{thm:policy}, which provides a family of asymptotically optimal policies. We first prove a general sufficient criterion for the admissibility of a certain class of trading strategies. It implies admissibility of our asymptotic optimizers and also guarantees the admissibility of the constant coefficient portfolios used to establish a lower bound for the value function in Proposition~\ref{local-bound-ratios}. To cover both of these applications, we consider the following class of feedback trading rates,
\begin{align}\label{ansatz1}
\dot H^{\e,\a,\b}(\e d,h,y)=-\frac{\a(y)}{\e} (h-\cM(y))-\frac{\b(y)}{\e}{d},
\end{align}
where $\cM$ is the frictionless Merton portfolio for the problem without illiquidity and $\a(y),\b(y)$ are $\mathbb{R}^{n\times n}$-valued functions.

\begin{prop}\label{prop:addmisibility}
Suppose that there exists a mapping $y\in \cD\to M(y)$ with values in $\Sb_{2n}$ such that, for some $\d>0$,
\begin{gather}
\notag M\geq \d I_{2n},\quad \cL^YM\leq (\rho-\d)M,\\
\notag M \left( \begin{array}{cc}
RI_n+C \b &  C\a  \\
 \b& \a   \end{array} \right)^\top+\left( \begin{array}{cc}
RI_n+C \b &  C\a  \\
 \b & \a   \end{array} \right)M \geq 0,\\
\E\left[\int_0^\infty e^{-(\rho-\d)t }\left(\mathrm{Tr}(M_t c_{\bar \cM}(Y_t))+\sum_{i,j=1}^{2n}\left|\frac{d\langle M^{i,j},\bar \cM^i\rangle_t}{dt}\right|^2\right)dt\right]<\infty,\label{cont-crossterms}
\end{gather}
{where } $\bar\cM_t = \begin{pmatrix} 0 \\ \cM_t \end{pmatrix}.$
Then for all $\e>0$ the feedback control \eqref{ansatz1} is admissible. In particular the controls \eqref{eq:asy-H} are admissible under the assumptions of Theorem \ref{thm:policy}. 
\end{prop} 

\begin{proof}
Recall the rescaled price distortion $\tilde D$ from \eqref{def:tildeD} and note that, for fixed $\e>0$, checking the transversality conditions \eqref{eq:trans} for $D$ or $\tilde D$ is equivalent. Define $X_t=(\tilde D^\top_t,H^\top_t -\cM_t)^\top\in \R^{2n}$ and 
$$N(y):=\left( \begin{array}{cc}
RI_n+C \b(y) &  C\a(y)  \\
 \b(y) & \a(y)   \end{array} \right).$$
With this notation, 
 $$dX_t:=-\frac{1}{\e}N_t X_t dt-d\bar\cM_t.$$
 It\^o's formula yields 
\begin{align}\label{contr-drift}
d\left(e^{-(\rho-\d/2)t}X_t^\top M_t X_t\right)&=e^{-(\rho-\d/2)t}\left( \mathrm{Tr}(M_t c_{\bar \cM}(Y_t))+X_t^\top \chi_t\right)dt\\
&\quad-\frac{e^{-(\rho-\d/2)t} }{\e}X_t^\top (M_tN_t+N_t^\top M_t)X_t dt\notag\\
&\quad+e^{-(\rho-\d/2)t} X_t^\top (\cL^Y M_t-(\rho-\d/2) M_t )X_t dt,\notag
\end{align}
up to local martingale, where $\chi$ can be explicitly written using $\left\{\frac{d\langle M^{i,j},\bar \cM^i\rangle_t}{dt}\right\}_{i,j}$ and whose moments can be bounded using the terms in \eqref{cont-crossterms}.  Taking into account the elementary estimate $-\frac{\d^2}{2} |X_t|^2 +\chi_t^\top X_t\leq \frac{1}{2\d^2} |\chi_t|^2$, it follows that
$$e^{-(\rho-\d/2)t}X_t^\top M_t X_t-\int_0^t e^{-(\rho-\d/2)s}\left( \mathrm{Tr}(M_s c_{\bar \cM}(Y_s)) + \frac{1}{2\d^2} |\chi_s|^2\right) ds$$
is a local supermartingale. In view of \eqref{cont-crossterms}, this proces is bounded from below by an integrable process, so that it is a true supermartingale and therefore converges to a finite limit almost surely and in $L^1$ as $t \to \infty$. As the process $\int_0^t e^{-(\rho-\d/2)s}\left( \mathrm{Tr}(M_s c_{\bar \cM}(Y_s)) + \frac{1}{2\d^2} |\chi_s|^2\right) ds$ is increasing and integrable by \eqref{cont-crossterms}, it follows that $e^{-(\rho-\d/2)t}X_t^\top M_t^\top X_t$ admits a finite limit as well. Therefore,
$$e^{-\rho t}X_t^\top M_t X_t\to 0 \quad \mbox{ and } \quad \E\left[\int_0^\infty e^{-\rho t}X_t^\top M_t X_t dt\right]<\infty,$$
so that the control is admissible as claimed.

To apply this result to the policies from Theorem~\ref{thm:policy}, let 
$$\a=\Lambda^{-1}Q_h^\top\mbox{ and }\b=\Lambda^{-1}Q_d^\top.$$
We now show that the conditions of the present proposition are satisfied for  $M=A$. As $A$ satisfies the Riccati equation \eqref{eq:matrixriccati}, we have
\begin{align}\label{eq:contractivity}
&A\left( \begin{array}{cc}
RI+C \Lambda^{-1}Q_d^\top &  C\Lambda^{-1}Q_h^\top \\
 \Lambda^{-1}Q_d^\top&\Lambda^{-1} Q_h^\top    \end{array} \right)+\left( \begin{array}{cc}
RI+C \Lambda^{-1}Q_d^\top  &  C\Lambda^{-1}Q_h^\top   \\
 \Lambda^{-1}Q_d^\top & \Lambda^{-1}Q_h^\top    \end{array} \right)^\top A \notag \\
 &=-\Gamma   A - A \Gamma +  2A  \hat C \Lambda^{-1}\hat C^\top  A \geq\Psi \geq 2\underline m I_{2n}.
\end{align}
Under Assumption~\ref{assume-admis} and the other conditions of Theorem~\ref{thm:policy} and by Lemma \ref{lem:ricatti},  the matrix $A$ is positive, and satisfies $\cL^YA\leq (\rho-\d_1)A$
as well as \eqref{cont-crossterms}. Whence, the controls \eqref{eq:asy-H} are admissible by the first part of the present proposition.
\end{proof}

We now turn to the proof of the asymptotic optimality of the policies proposed in Theorem~\ref{thm:policy}:

\begin{proof}[Proof of Theorem \ref{thm:policy}]
Recall the approximate value function $\hat V^\e$ defined in \eqref{eq:defhatV}. As in Section~\ref{ss.ansatz}, starting from $\theta=(\e d,h,y)$ we denote by $(\tilde D^\e,H^\e,Y)$ the state controlled with the feedback strategy $\dot H^\e$ from \eqref{eq:asy-H}. Define $X_t^\e=((\tilde D^\e)^\top_t,(H^\e)^\top_t -\cM_t)^\top\in \R^{2d}$ and 
$$N :=\left( \begin{array}{cc}
RI+C\Lambda^{-1} Q_d^\top  &  C\Lambda^{-1}Q_h^\top    \\
 \Lambda^{-1} Q_d^\top  &\Lambda^{-1}Q_h^\top    \end{array} \right).$$
With this notation, $dX_t^\e:=-\e^{-1} N_t X_t^\e dt-d\bar\cM_t$. It\^o's formula, applied to $ e^{-\rho T}\hat V^\e(\e \tilde D_T^\e,H_T^\e,Y_T)$, $T>0$,  and \eqref{eq:truemart} in turn yield
\begin{align*}
 &\hat V^\e(\e d,h,y)=\E\left[ e^{-\rho T}\hat V^\e(\e \tilde D^\e_T,H^\e_T,Y_T)-\int_0^T e^{-\rho t} \left(-\rho V^0_t +\cL^YV^0_t+\e\rho \varpi(X^\e_t,Y_t)\right.\right.\\
 &\quad\left.\left.-\e(-\rho u_t +\cL^Y u_t +a_t)+(H^\e_t)^\top C\Lambda^{-1}(Q_d^\top \tilde D^\e_t+Q_h^\top H^\e_t)+R(H^\e_t-C^{-1}\tilde D^\e_t)^\top \tilde D^\e_t\right.\right.\\
 &\quad\left.\left.+\frac{1}{2}(X^\e_t)^\top (\Psi_t +A_t\hat C\Lambda^{-1}\hat C^\top A_t)X^\e_t-\frac{\e}{2} (X^\e_t)^\top \chi_t-\frac{\e}{2}(X^\e_t)^\top\cL^YA_t X^\e_t\right)dt\right]\\
& \leq \E\left[ e^{-\rho T}\hat V^\e(\e \tilde D^\e_T,H^\e_T,Y_T)-\int_0^T e^{-\rho t}\left(-\frac{\e}{2} (X^\e_t)^\top \chi_t+ (H^\e_t)^\top C\Lambda^{-1}(Q_d^\top \tilde D^\e_t+Q_h^\top H^\e_t)\right.\right.\\
 &\left.\left.\quad+R(H^\e_t-C^{-1}\tilde D^\e_t)^\top \tilde D^\e_t+\frac{1}{2}(X^\e_t)^\top (\Psi_t +A_t\hat C\Lambda^{-1}\hat C^\top A_t)X^\e_t-\frac{\mu_t^\top\Sigma^{-1}_t\mu_t}{2\gamma}\right)dt\right].
\end{align*}
Here, $\chi$ is defined as in \eqref{contr-drift} and we have used the second corrector equation~\eqref{corrector2} satisfied by $u$, the Riccati equation \eqref{eq:matrixriccati} for $A$, the frictionless dynamic programming equation~\eqref{eq:V0} for $V^0$, as well as Assumption \ref{assume-admis}. Set
$$\beta_\e:=\int_0^\infty e^{-\rho t}\E\left[\frac{1}{2}| (X^\e_t)^\top \chi_t|\right] dt,$$
which satisfies
$$\beta_\e^2\leq c\int_0^\infty e^{-\rho t}\E\left[\frac{1}{2}| X^\e_t|^2\right]dt\int_0^\infty e^{-\rho t}\E\left[|\chi_t|^2\right] dt$$
for some constant $c>0$. By Assumption \eqref{eq:controlstate}, the right-hand side of this inequality tends to $0$ as $\e \to 0$, so that $\beta_\e\to 0$ as $\e\to 0$. 

As a consequence:
\begin{align*}
\hat V^\e(\e d,h,y)&\leq \E\left[ e^{-\rho T}\hat V^\e(\e \tilde D^\e_T,H^\e_T,Y_T)\right]+\e\beta_\e\\
&\quad+\E\left[\int_0^T e^{-\rho t}\left(  (H^\e_t)^\top \left(-R\tilde D^\e_t+C_\e\dot H_t^\e\right)+R (\tilde D^\e_t)^\top C^{-1} \tilde D^\e_t \right.\right.\\
&\left.\left.\quad\qquad\quad-\frac{1}{2}(X^\e_t)^\top (\Psi_t +A_t\hat C\Lambda^{-1}\hat C^\top A_t)X^\e_t+\frac{\mu_t^\top\Sigma^{-1}_t\mu_t}{2\gamma}\right)dt\right]\\
&=\E\left[ e^{-\rho T}\hat V^\e(\e \tilde D^\e_T,H^\e_T,Y_T)\right]+\e\beta_\e\\
&+\E\left[\int_0^T e^{-\rho t}\left(  (H^\e_t)^\top \left(-R \tilde D^\e_t+C\dot H^\e_t\right) -\frac{\gamma}{2}(H^\e_t-\cM_t)^\top\Sigma_t (H^\e_t-\cM_t)\right.\right.\\
&\left.\left.\quad\qquad\quad-\frac{1}{2}(X^\e_t)^\top ( A_t\hat C\Lambda^{-1}\hat C^\top A_t)X^\e_t+\frac{\mu_t^\top\Sigma^{-1}_t\mu_t}{2\gamma}\right)dt\right].
\end{align*}
Recalling that $\dot H_t^\e=-\e^{-1}\Lambda^{-1}\hat C^\top A_t X^\e_t$ and in turn
$$\frac{1}{2}(X_t^\e)^\top ( A_t\hat C\Lambda^{-1}\hat C^\top A_t)X_t^\e=\frac{1}{2}(\dot H^\e)^\top\Lambda_\e \dot H^\e,$$
we finally obtain
\begin{align*}
&\hat V^\e(\e d,h,y)\leq \lim_{T\to \infty}\E\left[ e^{-\rho T}\hat V^\e(\e \tilde D^\e_T,H^\e_T,Y_T)\right]+\e\beta_\e\\
&\quad+\E\left[\int_0^\infty e^{-\rho t}\left(  (H^\e_t)^\top \left(\mu_t-R \tilde D^\e_t+C\dot H^\e_t\right) -\frac{\gamma}{2}(H^\e_t)^\top\Sigma_t H^\e_t-\frac{1}{2}(\dot H_t^\e)^\top\Lambda_\e \dot H_t^\e\right)dt\right]\\
&\qquad\qquad\quad\leq \cJ^\e(d\e,h,y;\dot H^\e)+ \e\beta_\e= cJ^\e(d\e,h,y;\dot H^\e)+o(\e).
\end{align*}
In view of the value expansion in Theorem~\ref{thm:expansion}, the trading rates $(\dot H^\e)_{\e >0}$ therefore indeed are asymptotically optimal as claimed.
\end{proof}

\appendix
\section{Appendix: Additional Proofs}\label{sec:app}

\subsection{Additional Proofs for Section~\ref{sec:model}}

The following result shows that the comparison Assumption~\ref{assume-comparison} for the second corrector equation is satisfied for the model with mean-reverting returns from Example~\ref{ex:meanrev}.

\begin{prop}\label{prop:comparison-example}
Fix a constant $a$. Comparison holds for the PDE
\begin{align}\label{pde:example}
\rho u- \cL^Y u=\rho u -\nu (y_2) \pa_1 u +\lambda y_2 \pa_2 u -\frac{1}{2}\left(\sigma^2 \pa_{11}u+\eta^2 \pa_{22}u\right)=a,
\end{align}
among semisolutions $\phi: \R^2 \to \R$ satisfying the following growth condition,
\begin{align}\label{eq:growthappendix}c(1+|y_2|^2)^{2} \geq \phi(y) \geq  0, \quad \mbox{ for some }c>0.
\end{align}
\end{prop}

\begin{proof}
The first step is to exhibit a supersolution of the equation that dominates the semisolutions satisfying the growth condition for sufficiently large arguments. To this end, let
$$\hat \phi(y):=\phi_1(y_1)+\phi_2(y_2):=(c_2+|y_1|^2)^{1/2}+c_3(c_4+|y_2|^2)^{5/2}.$$
Note that for all function $\phi$ satisfying \eqref{eq:growthappendix}, there exists a compact set $K$ such that on the complement of $K$, $\hat \phi\geq \phi$ holds. A computation shows
\begin{align*}
\rho\hat \phi- \cL^Y \hat\phi&=\phi_1(y_1)\left(\rho -\frac{y_1\nu (y_2)}{2(c_2+|y_1|^2)}-\frac{\sigma^2}{2}\left(\frac{y_1^2}{4(c_2+|y_1|^2)^2}-\frac{1}{2(c_2+|y_1|^2)}\right)\right)\\
&\quad+\phi_2(y_2)\left(\rho +\frac{5 \lambda y_2^2}{(c_4+|y_2|^2)}-\frac{\eta^2}{2}\left(\frac{ 15 y_2^2}{4(c_4+|y_2|^2)^2}-\frac{5}{2(c_4+|y_2|^2)}\right)\right).
\end{align*}
By choosing sufficiently large $c_4,c_2$, we can control the terms coming from the second derivatives by $\rho/2$ and obtain 
\begin{align*}
\rho\hat \phi- \cL^Y\hat \phi&\geq \phi_1(y_1)\left(\frac{\rho}{2} -\frac{ y_1\nu(y_2)}{2(c_2+|y_1|^2)}\right)+\phi_2(y_2)\left(\frac{\rho}{2} +\frac{5 \lambda y_2^2}{(c_4+|y_2|^2)}\right)\\
&\geq \frac{\rho}{2}\hat\phi(y) -\nu(y_2)  \frac{  y_1}{2(c_2+|y_1|^2)^{1/2}}\geq\frac{\rho}{2}\hat\phi(y) -\frac{|\nu(y_2)|  }{2}.
\end{align*}
We now use the boundedness of the derivative of $\nu$ and take $c_3\geq 1$ large enough to obtain 
\begin{align*}
\rho\hat \phi- \cL^Y\hat \phi&\geq \frac{\rho}{2}\hat\phi(y) -\frac{|\nu(y_2)|  }{2} \geq  \frac{\rho}{4}\hat\phi(y).
\end{align*}
By choosing $c_3\geq 1$ sufficiently large we can therefore guarantee that the right-hand side dominates $a$ and therefore indeed is a supersolution of \eqref{pde:example}. 

To use this supersolution to establish comparison, argue by contradiction. If comparison does not hold, there is a subsolution $u$ and supersolution $v$ of the equation such that 
$\sup_{\R^2} \{u-v\} >0$. One can then find $\e>0$ small enough such that $v^\e:=(1-\e)v+\e\hat \phi$ (with the supersolution $\hat\phi$ constructed above) is also a supersolution of the equation satisfying $\sup_{\R^2} \{u-v^\e\} >0$.

By the growth conditions for $u,v$, and $\phi$, there exists a compact $K\subset \R^2$ such that 
$$u(y)-v^\e(y)\leq 0, \quad \mbox{for all $y\notin K$.} $$
The ``doubling-of-variable method'' and \cite[Theorem 3.2]{CrIsLi92} (applied on $K$) can now be used to obtain a contradiction.
\end{proof}

Next, we sketch how the weak-dynamic programming approach of \cite{bouchard.touzi.11} enables us to derive the viscosity property of the frictional value function.

\begin{proof}[Proof of Proposition \ref{asm:frictional.pde}]
The proof is a minor modification of \cite[Corollary 5.6]{bouchard.touzi.11}, also compare \cite[Proof of Theorem 2.1]{altarovici.al.15}. By \cite[Remark 3.11]{bouchard.touzi.11}, for all families $\{\tau^{\dot H}:\dot H\in \cA\}$ of uniformly bounded stopping times and upper semicontinuous minorants $\phi$ of $V$, the function $V$ satisfies the following weak dynamic programming principles,
\begin{align}
&V(\theta)\leq \sup_{\dot H\in\cA_\rho}\E\left[ e^{-\rho \tau^{\dot H}}V^*(D_{\tau^{\dot H}}^{\theta,\dot H},H_{\tau^{\dot H}}^{\theta,\dot H},Y^y_{\tau^{\dot H}})\right.\\
&\left.+\int_0^{\tau^{\dot H}}\e^{-\rho t} \left((H_t^{\theta,\dot H}) ^{\top}(\mu_t-R D_t^{\theta,\dot H}+C \dot{H}_t)-\frac{\gamma}{2} (H_t^{\theta,\dot H})^\top \Sigma_t H_t^{\theta,\dot H}-\frac{1}{2}\dot{H}^\top_t \Lambda \dot{H}_t\right) dt\right],\notag\\
&V(\theta)\geq \sup_{\dot H\in\cA_\rho}\E\left[ e^{-\rho \tau^{\dot H}}\phi(D_{\tau^{\dot H}}^{\theta,\dot H},H_{\tau^{\dot H}}^{\theta,\dot H},Y^y_{\tau^{\dot H}})\right.\\
&\left.+\int_0^{\tau^{\dot H}}\e^{-\rho t} \left((H_t^{\theta,\dot H}) ^{\top}(\mu_t-R D_t^{\theta,\dot H}+C \dot{H}_t)-\frac{\gamma}{2} (H_t^{\theta,\dot H})^\top \Sigma_t H_t^{\theta,\dot H}-\frac{1}{2}\dot{H}^\top_t \Lambda \dot{H}_t\right) dt\right].\notag
\end{align}
As the generator of \eqref{PDE} is continuous, one can now modify the proof of \cite[Corollary 5.6]{bouchard.touzi.11} to establish that $V^*$ (resp.~$V_*$) is a viscosity subsolution (resp.~supersolution) of the frictional dynamic programming equation~\eqref{PDE}. This is the definition of discontinuous viscosity property for $V$. 
\end{proof}

Next, we turn to the sufficient conditions for the finiteness of the frictional value function.

\begin{lem}\label{lem:conc}
Suppose Assumption \ref{finite-V0} is satisfied and there exists $\d>0$ such that for all 
$$M\in\cS:=\left\{ \left( {\begin{array}{cc}
   -(2R+\rho)C^{-1} & \rho I_n \\       \rho I_n & -\gamma\Sigma(y) \      \end{array} }\right) \in \Sb_{2n}
:y\in \cD\right\}$$
and $\xi\in \R^{2n}$ we have 
$$\xi^\top  M\xi\leq-2\d|\xi|^2.$$
Then, for all admissible trading rates $\dot{H} \in \cA_\rho$,
\begin{align}\label{eq:conc}
\notag &\cJ(d,h,y; \dot{H})-V^0(y)=-h^\top d+\frac{d^\top   C^{-1}d}{2}+\rho\E\left[\int_0^\infty e^{-\rho t}(H_t^{\theta,\dot H} )^\top D_t^{\theta,\dot H} dt\right]\notag\\
&\qquad-\E\left[\int_0^\infty e^{-\rho t}(D_t^{\theta,\dot H})^\top\left(\frac{(2R+\rho) C^{-1} }{2}\right)D_t^{\theta,\dot H}dt\right]\\
&\qquad-\E\left[\int_0^\infty \frac{e^{-\rho r}}{2}\left((H_t^{\theta,\dot H}-\cM_t)^\top\gamma\Sigma_t (H_t^{\theta,\dot H}-\cM_t)+\dot{H}_t^\top\Lambda_t\dot{H}_t \right)dt\right].\notag
\end{align}
This implies, in particular, that the expectation in \eqref{eq:control-e} is well defined for all $\dot H\in \cA_\rho$ and the frictional value function is finite, $V(\theta)<\infty$ for all $\theta=(d,h,y)\in \R^n\times \R^n\times \cD$.
\end{lem}

\begin{proof}

Fix an admissible control $\dot H\in\cA_\rho$. Note that due to our admissibility condition, the right-hand side of
\eqref{eq:conc} is well defined; we denote it by 
$\tilde \cJ (d,h,y; \dot{H}).$
Our objective is to use the admissibility condition to show that this quantity is equal to 
$$ \cJ(d,h,y; \dot{H})-V^0(y).$$
We differentiate $e^{-\rho r} D_r^\top C^{-1}D_r$ and use \eqref{eq:distortion} and the transversality condition~\eqref{eq:trans}, obtaining 
$$-\int_0^\infty e^{-\rho r}D_r^\top \left(\frac{2R+\rho}{2}C^{-1}\right) D_r dr+\frac{d^\top C^{-1}d}{2}=-\int_0^\infty e^{-\rho r}D_r ^\top \dot H_r dr.$$
Similarly we differentiate $e^{-\rho r}H_r^\top D_r$ to derive
\begin{align}\label{eq:blabla}
-\int_0^\infty e^{-\rho r}D_r^\top \dot H_r dr=h^\top d+\int_0^\infty e^{-\rho r}H_r^\top \dot D_rdr- \rho\int_0^\infty e^{-\rho r}H_r^\top D_rdr.
\end{align}
Another application of \eqref{eq:distortion} gives
\begin{align}
&-\int_0^\infty e^{-\rho r}D_r^\top \left(\frac{2R+\rho}{2}C^{-1}\right) D_r dr+\frac{d^\top C^{-1}d}{2}-h^\top d+\rho\int_0^\infty e^{-\rho r}H_r^\top D_rdr\\
&=-R \int_0^\infty e^{-\rho r}H_r^\top D_rdr+ \int_0^\infty e^{-\rho r}H_r^\top C \dot H_r dr.
\end{align}
Note that by definition of $\cM$ we have  
\begin{align}
&\frac{1}{2}(H_t^{\theta,\dot H}-\cM_t)^\top\gamma\Sigma_t (H_t^{\theta,\dot H}-\cM_t)=\frac{1}{2}H_t^\top \gamma\Sigma_t H_t^{\theta,\dot H} -\mu_t^\top H_t +\frac{\mu^\top_t,\Sigma^{-1}_t \mu_t}{2\gamma}
\end{align}
which implies finally that 
\begin{align}
&-\int_0^\infty e^{-\rho r}D_r^\top \left(\frac{2R+\rho}{2}C^{-1}\right) D_r dr+\frac{d^\top C^{-1}d}{2}-h^\top d+\rho\int_0^\infty e^{-\rho r}H_r^\top D_rdr\\
&-\int_0^\infty e^{-\rho r}\frac{1}{2}(H_r-\cM_r)^\top\gamma\Sigma_r (H_r-\cM_r)dr\\
&=\int_0^\infty e^{-\rho r}\left(H_r^\top (\mu_r -RD_r + C \dot H_r )-\frac{1}{2}H_r^\top \gamma\Sigma_r H_r -\frac{\mu^\top_r\Sigma^{-1}_r \mu_r}{2\gamma}\right)
dr.
\end{align}
Now, take the expectation of both sides and use \eqref{eq:vf} to obtain
$$\tilde \cJ(d,h,y; \dot{H})=\cJ (d,h,y; \dot{H})-V^0(y).$$
This shows that $\cJ$ is well defined for all $\dot H\in \cA_\rho$ and {\eqref{eq:conc}} holds.  We now rewrite the right-hand side of \eqref{eq:conc} as
\begin{align*}
-h^\top d&+\frac{d^\top C^{-1}d}{2}+\rho\E\left[\int_0^\infty e^{-\rho t}\cM_t^\top D_t dt\right]\\
&+\rho\E\left[\int_0^\infty e^{-\rho t}(H_t-\cM_t)^\top D_t dt\right]-\E\left[\int_0^\infty e^{-\rho t}D_t^\top\left(\frac{(2R+\rho) C^{-1}}{2}\right)D_t dt\right]\\
&-\E\left[\int_0^\infty \frac{e^{-\rho r}}{2}\left((H_t-\cM_t)^\top\gamma\Sigma_t (H_t-\cM_t)+\dot{H}_t^\top\Lambda_t\dot{H}_t \right)dt\right].
\end{align*}
Note that the last two lines in this expression correspond to the action of the matrices in $\cS$ on the vector $(D_t^\top, H_t^\top-\cM_t^\top)^\top$. By assumption, these are bounded from above by
   \begin{align}\label{eq:up-bound-lem}
   -\d (|D_t|^2 +|H_t-\cM_t|^2)
   \end{align}
Moreover,  together with the admissibility of $\dot H$, the $\e$-Young inequality yields 
\begin{align*}
&\rho\E\left[\int_0^\infty e^{-\rho t}\cM_t^\top D_t dt\right]=\rho\E\left[\int_0^\infty e^{-\rho t}((\gamma \Sigma_t)^{1/2}\cM_t)^\top((\gamma \Sigma_t)^{-1/2} D_t )dt\right]\\
&\qquad\leq \frac{ \rho^2}{(2\d)^2} \E\left[\int_0^\infty e^{-\rho t}\cM_t^\top \gamma\Sigma_t\cM_t  dt\right]+\d^2\E\left[\int_0^\infty e^{-\rho t} D_t^\top (\gamma\Sigma_t)^{-1} D_tdt\right]\\
&\qquad\leq \frac{ \rho^2}{(2\d)^2 } V^0(y)+\frac{\d}{2}\E\left[\int_0^\infty e^{-\rho t} |D_t|^2dt\right].
\end{align*}
   Notice that \eqref{eq:up-bound-lem} allows us to bound the last term above because
   $$\frac{\delta |D_t|^2}{2} -\d (|D_t|^2 +|H_t-\cM_t|^2)\leq 0.$$
This finally gives 
   \begin{align*}
   \cJ(d,h,y; \dot H)-V^0(y)&\leq-h^\top   d+\frac{d^\top   C^{-1}d}{2}+ \frac{ \rho^2}{(2\d)^2 } V^0(y).
\end{align*}
Moreover, we obtain the upper bound
$$V(d,h,y) \leq-h^\top   d+\frac{d^\top   C^{-1}d}{2}+\left(1+\frac{ \rho^2}{(2\d)^2 } \right)V^0(y)$$
by taking the supremum over admissible controls. This concludes the proof because $V^0$ is finite by Assumption~\eqref{finite-V0}.  
 \end{proof}

\begin{rem} \label{rem:para}
 The mappings $\dot H \to H_t^{\theta,\dot H}$ and $\dot H\to D_t^{\theta,\dot H}$ are affine and under the assumptions of Lemma~\ref{lem:conc} the mapping $(d,h) \mapsto (d, h) M (d, h)^\top$ is concave for all $M\in\cS$. In order to compare our result with \cite[Lemma 1]{garleanu.pedersen.16}, assume that $\Sigma $ is constant. Then, Lemma \ref{lem:conc}
 provides a sufficient condition for the concavity of $\dot H\to \cJ(d,h,y; \dot{H})$, which in turn yields that the frictional optimizer is unique. This sufficient condition is the positivity of the symmetric matrix 
 $$\left( {\begin{array}{cc}
   (2R+\rho)C^{-1} &- \rho I_n \\       -\rho I_n & \gamma\Sigma \      \end{array} } \right)  \mbox{or, equivalently, of } 
\left( {\begin{array}{cc}
   (2R+\rho)\gamma \Sigma^{1/2}C^{-1}\Sigma^{1/2} &- \rho I_n \\      - \rho I_n & I_n\      \end{array} } \right).$$
This is satisfied in particular if 
   \begin{equation}\label{eq:suff}
   (2R+\rho)\gamma \rho^{-2}> || \Sigma^{-1/2}C\Sigma^{-1/2}||,
   \end{equation}
 which is a sharper sufficient condition than the one from \cite[Lemma 1]{garleanu.pedersen.16}. In particular, \eqref{eq:suff} holds for sufficiently small discount rates $\rho$.
  Note also that if $C$ goes to $0$ or $R$ goes to infinity, then this condition is satisfied. 
  \end{rem}

\subsection{Additional Proofs for Section~\ref{sec:main}}\label{s:app3}

We first establish the properties of the solution $A$ of the Riccati equation \eqref{eq:matrixriccati}:

 \begin{proof}[Proof of Lemma \ref{lem:ricatti}]
The matrix $\Gamma  -\hat C\hat C^\top$ only has strictly negative eigenvalues; thus by \cite[Definition 5]{klamka2016controllability}, $(\Gamma, \hat C)$ is stabilizable. As, moreover, $\Psi$ is symmetric positive definite, \cite[Theorem 2.1]{ran1988existence} shows that there exists a maximal solution $ A$ of the Riccati equation  \eqref{eq:matrixriccati} such that all eigenvalues of 
$\Gamma-\hat C\Lambda^{-1}\hat C^\top A$ are nonpositive. In addition, by \cite[Theorem 2.2]{ran1988existence}, $ A$ is symmetric positive definite.

For $\lambda,\d>0$, define
$$
K_{\d,\lambda}:=\left( \begin{array}{cc}
\d I_{2n}&0 \\
  0&-\lambda I_{2n}\end{array} \right).
  $$
Then, for sufficiently small $\d>0$ and sufficiently large $\lambda>0$,
\begin{align}\label{eq:inequalityric}
\left( \begin{array}{cc}
\Psi&\Gamma \\
  \Gamma&-\hat C\Lambda^{-1}\hat C^\top\end{array} \right)\geq K_{\d,\lambda}.
  \end{align}
  Denote $A_{\d,\lambda}$ the maximal solution of the Riccati equation $$-\d I_{2n}+\lambda A_{\d,\lambda}A_{\d,\lambda}=0.$$
\cite[Theorem 2.2(i)]{ran1988existence} and  the inequality \eqref{eq:inequalityric} imply 
$$A\geq  A_{\d,\lambda}>0.$$  In view of Assumption \ref{assume-expansion}(i), the choice of $\d,\lambda$ can be made uniformly for $y\in\cD$ so that the lower bound in \eqref{lower-bound-varpi} also holds uniformly on $\cD$. 
  \end{proof}
  
Next, we show that the probabilistic representation \eqref{eq:fku} indeed provides the unique solution of the second corrector equation under suitable assumptions:
  
   \begin{proof}[Proof of Lemma \ref{lem:feynmann}]
 We first note that the continuity of $\mu_Y,\sigma_Y$, and $a$ implies that the upper and lower semicontinuous envelopes of the generator of \eqref{corrector2} coincide. Thus, the (discontinuous) viscosity property of $u$, i.e., that $u^*$ is a viscosity subsolution and $u_*$ is a viscosity supersolution, can be established similarly as in the proof of Proposition \ref{asm:frictional.pde}. The comparison result from Assumption \ref{assume-comparison} in turn yields $u^*\leq u_*\leq u^*$. Hence $u$ is indeed the unique continuous solution of \eqref{corrector2} as claimed.
 \end{proof}

Finally, we establish that the model with mean-reverting returns from Example~\ref{ex:OU} satisfies all regularity condition required for the application of Theorem~\ref{thm:policy}. First, we consider Assumptions~\ref{assume-expansion} and \ref{assume-comparison} 

\begin{proof}[Proof of Example~\ref{ex:riccati}]
 For the model with mean-reverting returns from Example~\ref{ex:OU}, $A(\cdot)$ and $M_\Sigma(\cdot)$ are constants. A direct computation shows that thanks to the boundedness of the derivatives of $\nu$, there exists $c>0$ such that $|\cL^Y\cM(y)|^4\leq c(1+|y_2|^2)^2$. Together with the mean reversion of $Y^2$, this implies $M_\cM(y) \leq c (1+|y_2|^2)^2$ for some $c>0$. Thus, \eqref{eq:comparison-class} becomes  $c(1+|y_2|^2)^{2} \geq \phi(y) \geq  0$, for some $c>0$. A comparison result for \eqref{corrector2} under this growth condition is estsablished in Proposition \ref{prop:comparison-example}. Note that it is obvious here that the function 
\begin{align}\label{eq:solu}
u(y) =\frac{A_2\eta^2}{\gamma^2\sigma^4}E\left[\int_0^\infty e^{-\rho t } (\nu')^2(Y^{2,y}_t)dt \right]
\end{align}
 is a smooth solution of this PDE. Nevertheless, a comparison result for semicontinuous semisolutions of the second corrector equation as in Proposition \ref{prop:comparison-example} is necessary because the upper and lower semilimits defined in Section \ref{s.semilimits} can a priori only be chracterized as viscosity semisolutions satisfying the growth condition \eqref{eq:comparison-class}. Hence, the comparison result from Proposition~\ref{prop:comparison-example} is crucial to obtain a full characterisation of these semilimits and deduce that they in fact coincide. 
\end{proof}

To conclude, we discuss the other regularity conditions from Theorem~\ref{thm:policy}:

   \begin{proof}[Proof of Example~\ref{ex:policy}]
For the model with mean-reverting returns from Example~\ref{ex:OU}, the first integrability condition \eqref{eq:assume:} is clearly satisfied because $A$ is constant and the squared diffusion coefficient $ \cM$ of the frictionless Merton portfolio is bounded in this case. The limit \eqref{eq:controlstate} posits that the optimally controlled states converge to their frictionless counterparts $(0,\cM)$ as the frictions vanish for $\e\to 0$. This can be verified as in the proofs of Lemma~\ref{expansion} and Proposition~\ref{prop:expansion} for any model where $A$ is constant as in Example~\ref{ex:OU}. 

In order to check \eqref{eq:cvhatV} one can explicitly compute each term in the definition of $\hat V^\e$ and obtain that a sufficient condition for \eqref{eq:cvhatV} is that, for all $\e>0$,
\begin{align}
\E\left[e^{-\rho T}\left(| Y^{2,y}_T|^2+| D_T^{\e d,h,y,\dot H^\e}|^2+|H_T^{\e d,h,y,\dot H^\e}-\cM_T|^2\right)\right]\to 0,\mbox{ as }T\to \infty.
\end{align}
In the context of Example~\ref{ex:OU}, the function $T\to \E\left[| Y^{2,y}_T|^2\right]$ is bounded. Set 
\begin{align}\label{def:N}N :=\left( \begin{array}{cc}
RI+C\Lambda^{-1} Q_d  &  C\Lambda^{-1}Q_h    \\
 \Lambda^{-1} Q_d  &\Lambda^{-1}Q_h    \end{array} \right)\end{align}
 which satisfies  $AN+N^\top A\geq \Psi$ by \eqref{eq:matrixriccati}.
 Then, in matrix-vector notation, the corresponding state dynamics are
 $$dX_t:=d\left( \begin{array}{c}
\e^{-1} D_t \\
 H_t-\cM_t \end{array} \right)=-\frac{N}{\e} X_t dt-d\bar\cM_t, \quad \mbox{where } \bar \cM:=\left( \begin{array}{c}
0 \\
 \cM\end{array} \right).$$
It\^o's formula applied to $X^\top_tAX_t$ shows
\begin{align*}
\frac{1}{2}d\E\left[X^\top_tAX_t\right]&=-\frac{1}{\e}\E\left[X^\top_t\left(NA+N^\top A)\right)X_t\right]dt\\
&-\lambda \E\left[(H_t^{\e d,h,y,\dot H^\e}-\cM_t)Y^{2,y}_t\right]dt +\eta^2 A_2 dt.
\end{align*}
A direct computation using the Riccati equation \eqref{eq:matrixriccati} shows that $NA+N^\top A\geq \Psi>0$.  The $\e$-Young inequality in turn yields that, for some constant $C_\e>0$,
$$\left|\lambda \E\left[(H_t^{\e d,h,y,\dot H^\e}-\cM_t)Y^{2,y}_t\right]\right|\leq \frac{1}{2\e}\E\left[X^\top_t\Psi X_t\right] +C_\e \E\left[|Y^{2,y}_t|^2\right].$$
As a consequence, $d\E\left[X^\top_tAX_t\right]\leq  2 C_\e \E\left[|Y^{2,y}_t|^2\right]dt+2\eta^2 A_2 dt$, so that $\E\left[X^\top_tAX_t\right]$ has at most linear growth in $t$; in particular, due to the uniform lower bound~\eqref{lower-bound-varpi} for $A$, \eqref{eq:cvhatV} holds.

To verify the martingale property in \eqref{eq:truemart}, recall that the frictionless value function $V^0$ is given by \eqref{eq:frictval} and its derivative in $y^2$ 
$$\pa_{y^2}V^0(y)=\frac{1}{\gamma\sigma^2 }E\left[\int_0^\infty e^{-(\rho+\lambda) t} (\nu\nu^\prime)(Y^{2,y}_t)dt\right]$$ has at most linear growth in $y^2$, and the components of $A$ are constant. Recall also the function $u=u(y)=u(y^2)$ defined in \eqref{eq:solu} whose derivative in $y^2$ can be proven to be bounded. The local martingale part in the It\^o decomposition of the approximate frictional value function $\hat{V}^\e$ from \eqref{eq:exp} therefore is 
\begin{align*}
&\int_0^\cdot \partial_{y^2} (V^0(Y_t)-\e\pa_{y^2}u(Y_t)) \eta dW^2_t\\
&+\frac{2\e\nu}{\gamma\sigma^2} \int_0^\cdot \left(A_{12} D^{\e d,h,y,\dot H^\e}_t+A_2 (H_t^{\e d,h,y,\dot H^\e}-\cM_t)\right) \eta dW_t^2.
\end{align*}
 Whence, the required square integrability of the integrands is a special case of \eqref{eq:controlstate} combined with the estimates in the second moment of $Y_t$.
\end{proof}

\bibliographystyle{agsm}
\bibliography{references}

\end{document}